\documentclass[aps,pra,superscriptaddress,twocolumn,nofootinbib,a4paper,amsfonts,amssymb,amsmath,floatfix]{revtex4-2}
\pdfoutput=1
\usepackage[utf8]{inputenc}   
\usepackage[american]{babel}  
\usepackage[colorlinks=true,allcolors=magenta,breaklinks=true]{hyperref}
\usepackage{microtype}	      
\usepackage{subcaption}	      
\usepackage{mathtools}	      
\usepackage{amsthm}	          
\usepackage{bm}		            
\usepackage{tikz}	            
\usetikzlibrary{calc}
\tikzset{point/.style={fill,draw,circle,minimum size=2,inner sep=0pt,outer sep=0pt}}
\tikzset{vertex/.style={point,minimum size=5}}
\tikzset{ccc/.style={vertex,fill=white}}
\tikzset{color1/.style={vertex,fill=red}}
\tikzset{color2/.style={vertex,fill=lime}}
\tikzset{box/.style={draw,fill=black,minimum width=1.2cm, minimum height=.5cm,inner sep=0pt,outer sep=0pt}}
\tikzset{nulllike/.style={-,dashed}}
\DeclareMathOperator{\An}{An}        
\DeclareMathOperator{\Co}{Co}		     
\newcommand{\ket}[1]{|#1\rangle}     
\DeclareMathOperator{\Pa}{Pa}				 
\DeclareMathOperator{\Ch}{Ch}				 
\DeclareMathOperator{\DAG}{DAG}			 
\DeclareMathOperator{\AC}{AC}				 
\DeclareMathOperator{\conv}{conv}		 
\DeclareMathOperator{\aff}{aff}			 
\DeclareMathOperator{\ext}{ext}			 
\DeclareMathOperator{\ord}{ord}			 
\newtheorem{theorem}{Theorem}		 		 
\newtheorem{lemma}[theorem]{Lemma} 	 
\theoremstyle{definition}				     
\newtheorem{definition}{Definition}  

\begin{document}
\title{The Möbius game and other Bell tests for relativity}
\author{Eleftherios-Ermis Tselentis}
  \affiliation{QuIC, Ecole Polytechnique de Bruxelles, C.P.\ 165, Université Libre de Bruxelles, 1050 Brussels, Belgium}
\author{\"Amin Baumeler}
	\affiliation{Facolt\`a di scienze informatiche, Universit\`a della Svizzera italiana, 6900 Lugano, Switzerland}
	\affiliation{Facolt\`a indipendente di Gandria, 6978 Gandria, Switzerland}

\begin{abstract}
	\noindent
	We derive multiparty games that, if the winning chance exceeds a certain limit, prove the incompatibility of the parties' causal relations with any partial order.
	This, in turn, means that the parties exert a back-action on the causal relations; the causal relations are {\em dynamical.}
	The games turn out to be representable by directed graphs, for instance by an orientation of the {\em M\"obius ladder.}
	We discuss these games as device-independent tests of spacetime's dynamical nature in general relativity.
	To do so, we design relativistic settings where, in the Minkowski spacetime, the winning chance is bound to the limits.
	In contrast, we find otherwise tame processes with classical control of causal order that win the games deterministically.
	These suggest a violation of the bounds in gravitational implementations.
	We obtain these games by uncovering a ``pairwise central symmetry'' of the correlations in question.
	This symmetry allows us to recycle the facets of the {\em acyclic subgraph polytope\/} studied by Gr\"otschel, J\"unger, and Reinelt in the mid-80s for combinatorial optimization.
	In addition, we derive multiparty games in a scenario where the polytope dimension grows only linearly in the number of parties.
	Here, exceeding the limits not only proves the dynamical nature \mbox{of the causal relations, but also that the correlations are incompatible with {\em any global causal order.}}
\end{abstract}

\maketitle
\section{Introduction}
Bell~\cite{bell1964}, in his seminal work, showed that quantum correlations are causally inexplicable~\cite{bell1975}.
Such {\em nonlocal\/} correlations arise as the spontaneous spacelike separated creation of fresh albeit strongly correlated data~\cite{wolf2017}.
Any tentative causal explanation of these correlations---without regressing to the formalism of quantum theory~\cite{lorenz2022}---necessitates an infinite speed of communication~\mbox{\cite{coretti2011,bancal2012,barnea2013,transitivequantum24}} and infinite precision~\cite{wood2015};
a clash with relativity.
In order to show his result, Bell employed a {\em device-independent technique.}
He derived limits on correlations---expressed via inequalities---from assumptions on the observed data only, without invoking any specifics of any theory.
If experimental observations exceed these limits {\it i.e.,} a Bell inequality is violated, then one is forced to reject one of the assumptions.
This is known as a  {\em Bell test.}
Bell inequalities are commonly expressed as collaborative multiparty games with a bound on the winning probability.
Bell's discovery is not only considered as one of the most fascinating in quantum theory, but also substantially boosted the development of quantum information~\cite{nielsenchuang2010,unspeakables2}:
The informational abstraction and simplification of quantum mechanics as a theory of qubits.
In this work, we extend Bell's reasoning to {\em relativity,}
and obtain Bell tests to certify the {\em dynamical nature of causal relations.}
In addition, this might---similar to the quantum case---aid in the development of a theory of ``gravity information.''

\subsection{Dynamical and indefinite causal order}
Quantum nonlocal correlations, as mentioned above, indicate the central role of {\em causality\/} in unifying quantum theory with relativity.
As pointed out by Hardy~\cite{hardy2005}, there is a complementary aspect that hints at the same.
On the one hand, general relativity is a {\em deterministic\/} theory equipped with a {\em dynamical\/} spacetime: The causal relations among events depend on the mass-energy distribution.
On the other, quantum theory is a {\em probabilistic\/} theory equipped with a {\em static\/} spacetime.
Therefore, it is expected that a theory of quantum gravity features {\em indefinite causal order,} {\it e.g.,} by extending quantum superposition to causal relations~\mbox{\cite{colnaghi2012,oreshkov2012,chiribella2013,portmann2017,zych2019}}.

Exemplary, indefinite causal order arises from the celebrated {\em quantum switch\/}~\cite{chiribella2013} (see Figure~\ref{fig:switch}).
Here, the causal relation between two experiments, Bob's and Charlie's, is coherently controlled by Alice's experiment in their common causal past.
\begin{figure}
	\centering
	\begin{tikzpicture}[out=90,in=270,looseness=1.25]
		\node[box] (A) at (0,0) {\color{white}\footnotesize\textbf{Alice}};
		\node[box] (B) at (-1,1.5) {\color{white} \footnotesize\textbf{Bob}};
		\node[box] (C) at (1,1.5) {\color{white} \footnotesize\textbf{Charlie}};
		\draw[->] (A.north) to (B.south);
		\draw[->] (B.north) to (C.south);
		\draw[->] (C.north) to (0,3);
		\draw[dashed,->] (A.north) to (C.south);
		\draw[dashed,->] (C.north) to (B.south);
		\draw[dashed,->] (B.north) to (0,3);
	\end{tikzpicture}
	\caption{The quantum switch: Depending on whether Alice prepares a control qubit in the state~$\ket 0$ or~$\ket 1$, a~target system~$\ket\psi$ traverses Bob's laboratory before Charlie's or {\em vice versa.} If Alice prepares the control qubit in a superposition state, then the trajectory of the target system is entangled with the control qubit: The causal relation between Bob and Charlie is {\em indefinite.}}
	\label{fig:switch}
\end{figure}
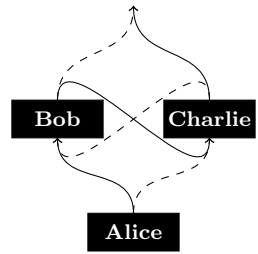
More specifically, Alice may perform an experiment such that Bob's experiment is performed {\em before\/} Charlie's, or another such that Bob's experiment is performed {\em after\/} Charlie's.
The causal relation between Bob's and Charlie's experiment is {\em dynamical.}
If Alice, however, performs both experiments in superposition, then {\em quantum indefiniteness\/} is injected into the causal relations, and dynamical causal order is turned into quantum indefinite causal order.
It is well-known~\mbox{\cite{araujo2015,oreshkov2016,wechs2021}} that no Bell test can certify the quantum indefinite causal order as exhibited by the isolated quantum switch and its generalizations~\cite{wechs2021}---unless additional assumptions are~\mbox{invoked~\cite{zych2019,tein2022,gogioso2023}}.
This is natural for the following reason.
In stark contrast to the Bell scenario~\cite{bell1964} where nonsignaling correlations are established, here, the limits of communication are of relevance.
But in unrestricted settings, all quantum communication can always be simulated classically.

\subsection{The M\"obius game}
In this work, we devise Bell tests for dynamical causal order.
Events obey a non-dynamical, or {\em static,} causal order whenever the parties' actions, which constitute the events, do not alter the causal relations.
In other words, events obey static causal order whenever the causal relations form a {\em partial order.}
A violation of the presented inequalities, that bound the correlations to respect a static causal order, implies then that the parties {\em altered\/} the causal relations.
These inequalities are---as we will see later---also violated by the quantum switch.
While, as mentioned above, the quantum indefinite nature of the causal relations in the isolated quantum switch cannot be certified, the {\em dynamical part indeed can.}

These inequalities are the facets of the partial-order-correlation polytope.
It turns out that the minimum of a~{\em single\/} output bit suffices to obtain such Bell tests.
In fact, the dimension of the polytope grows only quadratically as~$2n(n-1)$ in the number of parties~$n$.
We find that a projection of the polytope is the polytope of directed acyclic graphs (DAGs).
The latter has been studied by Gr\"otschel, J\"unger, and Reinelt~\cite{acyclic1985} in the mid-80s in the context of discrete optimization.
Moreover, we find that we can lift the facet-defining inequalities of this DAG polytope to obtain the relevant Bell games.
The key insight for this step is that the polytope of interest is what we call ``pairwise centrally symmetric;''
a symmetry that reflects the possibility of the parties to relabel the output bit.

An example of such a Bell test is represented by a~directed M\"obius ladder (see Figure~\ref{fig:mooebiusladder}).
\begin{figure}
	\centering
	\begin{tikzpicture}[every path/.style={-stealth,thick}]
		\def\k{5}
		\def\radius{1}
		\def\d{.5}
		\pgfmathtruncatemacro{\rot}{90+360/(2*\k)}
		\foreach \i in {1,2,...,\k} {
			\node[vertex] (v\i) at (\rot+360-\i*360/\k:\radius+\d) {};
			\node[vertex] (w\i) at (\rot+360-\i*360/\k:\radius) {};
		}
		\foreach \i in {1,3,...,\k}
			\draw (v\i) -- (w\i);
		\foreach \i in {2,4,...,\k}
			\draw (w\i) -- (v\i);
		\pgfmathtruncatemacro{\km}{\k-1}
		\foreach \i [evaluate=\i as \j using int(\i+1)] in {1,3,...,\km} {
			\draw (w\i) -- (w\j);
			\draw (v\j) -- (v\i);
		}
		\foreach \i [evaluate=\i as \j using int(\i+1)] in {2,4,...,\km} {
			\draw (w\j) -- (w\i);
			\draw (v\i) -- (v\j);
		}
		\draw (w\k) -- (v1);
		\draw (w1) -- (v\k);
	\end{tikzpicture}
	\caption{A directed M\"obius ladder~\cite{guy1967}. This graph represents a Bell game for ten or more parties with which dynamical causal order is detected.}
	\label{fig:mooebiusladder}
\end{figure}
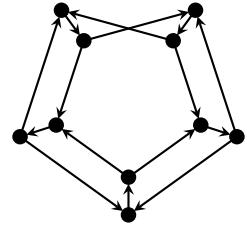
As a preview, the M\"obius game is the following.
Each vertex on the graph represents a party.
A referee picks at random an arc~$(s,r)$ from the graph, and a bit~$x$.
Then, the referee broadcast the chosen arc to all parties, and the bit~$x$ to the ``sender''~$s$ only.
The game is won whenever the ``receiver''~$r$ outputs~$x$.
If the winning probability exceed~$11/12$, then the parties' causal arrangement is {\em incompatible\/} with any partial order or mixtures thereof, and we must conclude that, if we want to retain a causal explanation, then the parties' actions {\em influenced\/} the causal relations.
Note that a violation of such an inequality could also be explained by abandonding the notion of a causal order, e.g., by the use of closed time-like curves.
In addition to the M\"obius game, we find other games represented by directed graphs.
The M\"obius game, however, has a clear advantage.
The bound of~$11/12$ holds for any finite number of parties involved.
In contrast, the bounds of the other games approach one for an increasing number of parties.
This favorable situation renders the tests reliable, and show that the feature of dynamical causal order is non-vanishing, as it is the case for~\mbox{nonlocality~\cite{collins2002,seevinck2002}} and noncausality~\cite{baumeler2020}.

\subsection{Relativity}
We apply our results to relativity.
Here, we find arrangements of events, defined by the {\em crossing of light beams,} such that the derived inequalities are never violated within special relativity.
If the games were played in a~gravitational setting, however, the dynamical spacetime of general relativity may be certified.
In this work, we do not carry out general-relativistic calculations, but support our claim via the description of classically well-behaved processes that win these games deterministically.
Moreover, these processes are conservative in the following sense:
An event can only influence its causal future, including the causal relations among the events within its causal future.
Such processes were recently investigated in detail due to their affinity for physical implementations~\cite{wechs2021}.
If a displacement of matter in general relativity alters the causal relations of events within the future lightcone, then these processes may be implemented in general relativity, and the back-action of matter to the spacetime structure can be certified.

\subsection{Outline}
We start with a description of our notation, and necessary basics on convex polytopes and graphs.
Thereafter, we provide the definition and derive properties of what we call {\em pairwise centrally symmetric polytopes.}
This symmetry refines central symmetry, and is our key-method for the projection and lifting of the polytopes (in particular, see Theorem~\ref{thm:liftingfacets}).
After that, we describe the correlations of interest, and compute the facet-defining Bell games (see Theorem~\ref{thm:bellgames}).
This is followed by a shorter part where the scenario is further simplified.
While the resulting polytope grows only linearly as~$2n$ in the number of parties~$n$, we show that the derived Bell game cannot distinguish between partial-order and {\em causal correlations\/} (see Theorem~\ref{thm:bellgamesimplified}).
The latter describes correlations where the parties may alter the causal relations of all parties in their future.
Then, we introduce the framework of {\em processes and causal models,} and show how these games are won deterministically therein (see Theorem~\ref{thm:winning}).
This is followed by a discussion of the results in the context of special and general relativity.
We end this work with conclusions and a series of open questions.

\section{Notation, polytopes, and graphs}
We use~$[n]$ for the set~$\{0,1,\dots,n-1\}$, and~$[n]^2_{\neq}$ for the set of all {\em distinct\/} pairs over~$[n]$, {\it i.e.,}~$[n]\times[n]\setminus\{(k,k)\}_{k\in[n]}$.
Lowercase letters will usually be used to express values, calligraphic letters sets, and bold letters matrices and vectors.
We simplify the common notation of~$P_{A\mid X}$ for a conditional probability distribution of~$A$ given~$X$, and~$P_{A\mid X}(a|x)$ for the conditional probability that~$A$ takes value~$a$ with~$X$ having value~$x$, by only referring to the probability-density function, {\it e.g.,}~$p(a|x)$.
We may always index the entries of a~$d$-dimensional vector~$\bm v$ with the natural numbers in~$[d]$, {\it i.e.,}~$\bm v=(v_i)_{i\in[d]}$.
The vector~$\bm 0$ is the all-zero vector,~$\bm e_i$ is~$\bm 0$ with a one in dimension~$i$, and~$\bm 1$ is the all-one vector.
The dimension of these vectors is always understood from the context.
The symbol~$\oplus$ is used for element-wise addition modulo two.
For a~collection~$\{\Delta_k\}_{k\in[n]}$ of objects labeled by~$[n]$ and a set~$\mathcal S\subseteq [n]$, we use~$\Delta_{\mathcal S}$ to denote the natural composition of the elements~$\{\Delta_k\}_{k\in\mathcal S}$, {\it e.g.,} the expression~$(a_k)_{k\in\mathcal S}\in\bigtimes_{k\in\mathcal S}\mathcal A_k$ may equally be written as~$a_{\mathcal S}\in\mathcal A_{\mathcal S}$.
We use the underlined symbol~$\underline\Delta$ whenever the set of labels is~$[n]$, and we may write~$\underline a\in\underline{\mathcal A}$.
Also, we use the subscript~\mbox{$\setminus \mathcal T$}, where~$\mathcal T\subseteq[n]$ is a set, for the natural composition over~$[n]\setminus\mathcal T$.
We further extend this notation in relation to a partial order~$\preceq_\sigma$ over~$[n]$:
The expression~$\Delta_{\preceq_\sigma k}$ for some~$k\in[n]$ denotes the natural composition over the elements in~$\{i\in[n] \mid i\preceq_\sigma k\}$, and similarly in the non-reflexive case~$\prec_\sigma$.

In this work, we employ the theory of convex polytopes~\cite{polytopes}.
A convex polytope~$\mathcal P\subseteq\mathbb R^d$ is the convex hull of a~finite set of vectors~$\mathcal S\subseteq\mathbb R^d$, denoted by~$\mathcal P = \conv(\mathcal S)$.
This is known as its~$V$-representation.
Equivalently,~$\mathcal P$ is the intersection of a finite set of halfspaces~$\mathcal P=\{\bm{x}\in\mathbb R^d\mid \bm A\bm x\leq \bm z\}$ for a set of~$m$ inequalities given by~$\bm A\in\mathbb R^{m\times d},\bm z\in\mathbb R^m$.
This is known as the~$H$-representation.
Since we always consider {\em convex\/} polytopes, we sometimes omit the term {\em convex\/} in the remaining of this work.
The dimension~$\dim(\mathcal P)$ of~$\mathcal P$ is the dimension of its affine hull~$\aff(\mathcal P)$.
The polytope~$\mathcal P\in\mathbb R^d$ is {\em full-dimensional\/} if it is~$d$-dimensional, {\it i.e.,} if its affine hull is the ambient space.
If~$\mathcal P$ is full-dimensional, then there exists a {\em unique\/} (up to multiplication)~$H$-representation.
A~vector~$\bm p\in\mathcal P$ is {\em extremal\/} if~$\bm p\not\in\conv(\mathcal P\setminus\{\bm p\})$.
The set of extremal vectors of~$\mathcal P$ is~$\ext(\mathcal P)$.
If~$\ext(\mathcal P)\subseteq\{0,1\}^d$, then~$\mathcal P$ is called a~{\em $0/1$ polytope.}
A linear inequality~\mbox{$\bm w\cdot\bm x\leq c$}, with~\mbox{$\bm w\in\mathbb R^d$} and~\mbox{$c\in\mathbb R$}, is {\em valid for~$\mathcal P$\/} whenever~$\forall \bm p\in\mathcal P: \bm w\cdot\bm p\leq c$.
We represent inequalities as pairs~$(\bm w,c)$.
An inequality~$(\bm w,c)$
is {\em trivial\/} if~$\bm w$ contains at most a single nonzero entry,
and {\em non-negative\/} if all entries of~$\bm w$ are non-negative.
If~$(\bm w,c)$ is valid for~$\mathcal P$, then~$\mathcal F:=\mathcal P\cap \{\bm p\in\mathbb R^d\mid\bm w \cdot\bm p =c\}$ is a {\em face\/} of~$\mathcal P$.
Faces of dimension~$\dim(\mathcal P)-1$ are called {\em facets,} and inequalities that give rise to facets are called {\em facet defining.}
If~$\bm c+\bm \delta\in\mathcal P \Leftrightarrow \bm c - \bm\delta\in\mathcal P$ for some center~$\bm c\in\mathcal P$, then~$\mathcal P$ is called {\em centrally symmetric.}

We also make extensive use of graph theory~\cite{bangjensen2009}.
A~graph~$G=(\mathcal V,\mathcal E)$ consists of a finite and non-empty set of vertices~$\mathcal V$ and a set of edges~\mbox{$\mathcal E\subseteq\{\{u,v\}\mid u,v\in \mathcal V\}$}.
For two graphs~$G$ and~$G'=(\mathcal V',\mathcal E')$, the {\em Cartesian product\/}~$G\times G'$ is the graph with the vertices~\mbox{$\mathcal V\times \mathcal V'$}, and~$\{(u,u'),(v,v')\}$ is an edge if and only if either~$u=v$ and~$(u',v')\in\mathcal E'$, or~$u'=v'$ and~$(u,v)\in\mathcal E$.
Each region bounded by the edges of a planar graph, {\it i.e.,} a graph drawn on the plane, is called a {\em face.}
A face surrounded by edges is called {\em internal face.}
A {\em directed\/} graph ({\em digraph\/} for short)~$D=(\mathcal V,\mathcal A)$ is a graph where the edges have a~direction, {\it i.e.,}~$\mathcal V$ is a finite and non-empty set of vertices, and~\mbox{$\mathcal A\subseteq \mathcal V\times \mathcal V$} is a set of ordered pairs called {\em arcs.}
The {\em order\/}~$\ord(G)$ of a graph (digraph) is the cardinality~$|\mathcal V|$ of its vertex set.
A {\em bipartite\/} graph (digraph) is two-colorable, {\it i.e.,} each vertex can be colored with one out of two colors such that adjacent vertices have different colors.
A~digraph~$D$ is an {\em orientation\/} of a graph~$G$ if there exists an ordering of the elements in each edge~$o(\mathcal E)$, such that~\mbox{$\mathcal A=o(\mathcal E)$}.
In a digraph, a {\em directed path\/} is a sequence of connected arcs~$((v_0,v_1),(v_1,v_2),\dots)$ where no vertex is revisited, and a {\em directed cycle\/} is a directed path where the last vertex coincides with the first.
A directed cycle with~$k$ arcs is a~{\em $k$-cycle.}
A digraph~$D$ is {\em simple\/} if it does not contain self-loops, {\it i.e.,}~$\forall v\in \mathcal V:(v,v)\not\in \mathcal A$.
A simple digraph is a {\em directed acyclic graph\/} (DAG for short), if it does not contain any directed cycles.
The set of all DAGs of order~$n$ is~$\DAG_n$, and the set of all DAGs is~$\DAG$.
If~$D'=(\mathcal V',\mathcal A')$ with~$\mathcal V'\subseteq \mathcal V$ and~$\mathcal A'\subseteq \mathcal A$, then~$D'$ is a {\em subdigraph\/} of~$D$, and we write~$D'\subseteq D$.
A simple digraph~$D=(\mathcal V,\mathcal A)$ of order~$n$ may be represented by its {\em adjacency vector\/}~$\bm \alpha(D):=(\alpha_{u,v})_{(u,v)\in [n]^2_{\neq}}\subseteq\{0,1\}^{n(n-1)}$, where~$\alpha_{u,v}$ is one if~$(u,v)\in\mathcal A$, and zero otherwise.

Basic graphs are the {\em line graph~$L_n$\/} with vertices~$[n]$ and edges~$\{\{i,i+1\}\mid i\in[n-1]\}$,
and the {\em complete bipartite graph~$K_{m,n}$\/} with vertices~$\{0\}\times [m] \cup \{1\}\times [n]$ and edges~$\{\{(0,u),(1,v)\}\mid u\in[m],v\in [n]\}$.
Basic digraphs are the {\em complete digraph~$K^{\text{\textnormal di}}_{n}$\/} with vertices~$[n]$ and arcs~$[n]^2_{\neq}$, and the~{\em $k$-cycle digraph~$C_k$} with vertices~$[n]$ and arcs~\mbox{$\{(i,i+1)\mid i\in[n-1]\}\cup\{(n-1, 0)\}$}.

\section{Pairwise central symmetry}
As key-method to establish our results, we introduce and study pairwise centrally symmetric 0/1 polytopes.
Pairwise central symmetry is a specific form of central symmetry in which the dimensions are paired, and the polytope has a center per pair.
For instance, if we pair the first two dimensions, then there exists a vector~$\bm c$ in the polytope such that~$\bm c+(x,y,\bm 0)$ is in the polytope if and only if~$\bm c-(x,y,\bm 0)$ is.
These polytopes contain a~``redundancy'' as the manifestation of relabelings.
\begin{definition}[Pairwise central symmetry]
	A 0/1 polytope~$\mathcal P\subseteq\mathbb R^{2d}$ is {\em pairwise centrally symmetric\/} if and only if
	there exists a permutation $\tau$ of the dimensions such that
	\begin{align}
		\begin{split}
			&\bm p\in \ext (\mathcal P) \implies\\
			\forall i\in[d]:\,&\bm p\oplus \bm e_{\tau^{-1}(i)}\oplus \bm e_{\tau^{-1}(d+i)} \in \ext (\mathcal P)
			\,.
		\end{split}
		\label{eq:pcs}
	\end{align}
	
	In the following, whenever we refer to a pairwise centrally symmetric polytope, we consider the isomorphic polytope where we reorder the dimensions such that  $\tau$ is the identity, and we define~$\bm p_0:=(p_i)_{i\in [d]}$, and~$\bm p_1:=(p_{d+i})_{i\in [d]}$.
	
\end{definition}
We start by showing two basic facts of such polytopes.
\begin{lemma}[Extension]
	\label{lemma:extension}
	The 0/1 polytope~$\mathcal P\subseteq\mathbb R^{2d}$ is pairwise centrally symmetric if and only if
	\begin{align}
		\begin{split}
			&\bm p = (\bm p_{0}, \bm p_{1}) \in \ext (\mathcal P) \Longleftrightarrow\\
			\forall \bm z\in\{0,1\}^d:\,&(\bm p_0\oplus \bm z,\bm p_1\oplus \bm z) \in \ext (\mathcal P)
			\,.
		\end{split}
		\label{eq:extension}
	\end{align}
\end{lemma}
\begin{proof}
	Iteratively apply Eq.~\eqref{eq:pcs}.
\end{proof}
\begin{lemma}[Central symmetry]
	If~$\mathcal P\subseteq\mathbb R^{2d}$ is a pairwise centrally symmetric 0/1 polytope, then it is also centrally symmetric.
\end{lemma}
\begin{proof}
	Let~$\bm c:=1/2\cdot\bm 1$ be the center,~\mbox{$\bm\delta\in\{\pm1/2\}^{2d}$}, and define the vectors~\mbox{$\bm p^{\pm}:=\bm c \pm \bm \delta\in\{0,1\}^{2d}$}.
	Now, we have the identity~$\bm p^+\oplus\bm 1 = \bm p^-$:
	Flipping all bits is equivalent to subtracting instead of adding~$\bm\delta$ to the center.
	Finally, by Lemma~\ref{lemma:extension} we have~\mbox{$\bm p^+\in\ext(\mathcal P)\Leftrightarrow \bm p^+\oplus\bm 1\in\ext(\mathcal P)$}.
\end{proof}

Pairwise centrally symmetric 0/1 polytopes have a natural projection.
\begin{definition}[Projection]
	\label{def:projection}
	Let~$\mathcal P\subseteq\mathbb R^{2d}$ be a pairwise centrally symmetric 0/1 polytope.
	We define the corresponding {\em projection map\/}~$\pi_d$ as
	\begin{align}
		\pi_d:\{0,1\}^{2d} &\rightarrow \{0,1\}^d\\
		(\bm v_0, \bm v_1) &\mapsto \bm v_0 \oplus \bm v_1
		\,.
	\end{align}
	The {\em projected polytope\/} is~\mbox{$\mathcal Q:=\conv(\pi_d(\ext(\mathcal P)))$}.
\end{definition}
This projection cancels the multitude of the extremal vectors which arises from the arbitrary vectors~$\bm z$ in Eq.~\eqref{eq:extension}.\footnote{Pairwise central symmetry can also be understood through the equivalence relation~$\bm p \sim \bm p'$ defined by~$\bm p_0\oplus \bm p_1 = \bm p'_0\oplus \bm p'_1$.
	The extremal vectors of the natural projection are then representatives of the equivalence classes.
}
In fact, all the ``information'' of~$\mathcal P$ is contained in the projection~$\mathcal Q$:
\begin{lemma}[Lifting polytope]
	\label{lemma:liftingpolytope}
	Define the {\em parametrized lifting map\/}~$\lambda_{\bm z}$ for~$\bm z\in\{0,1\}^d$ as
	\begin{align}
		\lambda_{\bm z}: \{0,1\}^d &\rightarrow \{0,1\}^{2d}\\
		\bm q &\mapsto (\bm q\oplus \bm z, \bm z)
		\,.
	\end{align}
	If~$\mathcal P\subseteq\mathbb R^{2d}$ is a pairwise centrally symmetric 0/1 polytope, and~$\mathcal Q$ the corresponding projected polytope,
	then~$\mathcal P$ is recovered from~$\mathcal Q$ via
	\begin{align}
		\ext(\mathcal P) = \bigcup_{\bm z \in \{0,1\}^d}\lambda_{\bm z}(\ext(\mathcal Q))
		\,.
	\end{align}
\end{lemma}
\begin{proof}
	By Definition~\ref{def:projection}, the extremal vectors of~$\mathcal Q$ are~$\pi_d(\ext(\mathcal P))$.
	To prove the~$\subseteq$ direction, let~\mbox{$\bm p=(\bm p_0,\bm p_1)$} be an element in~$\ext(\mathcal P)$.
	Now, by taking~$\bm z = \bm p_1$, we get that the right-hand side contains~\mbox{$\lambda_{\bm z} ( \pi_d(\bm p))=\lambda_{\bm p_1} (\bm p_0\oplus \bm p_1) = \bm p$}.
	For the reverse direction~$\supseteq$, let~$\bm r=(\bm r_0,\bm r_1)$ be an element of the right-hand side.
	This vector arises from some~\mbox{$\bm p=(\bm p_0,\bm p_1)\in\mathcal P$} and some~$\bm z$ through the identities~$\bm r_0=\bm p_0\oplus \bm p_1\oplus \bm z$ and~$\bm r_1=\bm z$.
	Due to Lemma~\ref{lemma:extension} and by setting~$\bm y:=\bm p_1\oplus \bm z$, the vector~$\bm r$ is also an element of the left-hand side because~$(\bm r_0\oplus \bm y,\bm r_1\oplus \bm y)$ is.
\end{proof}

Our main result on pairwise centrally symmetric polytopes is that the facets of the projected polytope can be {\em recycled:}
\begin{theorem}[Lifting facets]
	\label{thm:liftingfacets}
	Let~$\mathcal P\subseteq\mathbb R^{2d}$ be a pairwise centrally symmetric 0/1 polytope, and~$\mathcal Q$ the corresponding projected polytope.
	If~$\mathcal P$ and~$\mathcal Q$ are full-dimensional, and if~$(\bm w,c)$ is a non-negative and non-trivial facet-defining inequality of~$\mathcal Q$, then~$((\bm w,-\bm w),c)$ is a non-trivial facet-defining inequality of~$\mathcal P$.
\end{theorem}
\begin{proof}
	First, we show that the lifted inequality is {\em valid\/} for~$\mathcal P$.
	Thanks to Lemma~\ref{lemma:liftingpolytope}, each vector~$\bm p\in\ext(\mathcal P)$ can be expressed as~$\lambda_{\bm z}(\bm q)$ for some~$\bm q\in\mathcal Q$ and some~$\bm z$, thus
	\begin{align}
		(\bm w,-\bm w) \cdot (\bm q\oplus \bm z,\bm z)
		&=
		\sum_{i\in[d]}
		w_i(q_i\oplus z_i-z_i)
		\\
		&=
		\sum_{i\in[d]: z_i=0}
		w_iq_i
		+
		\Delta
		\,.
	\end{align}
	Since~$(\bm w,c)$ is non-negative, we have~$\Delta \leq 0$, and validity is inherited from the validity of~$(\bm w,c)$ for~$\mathcal Q$.

	In the following, we find~$2d$ affinely independent vectors in~$\ext(\mathcal P)$ that saturate the lifted inequality.
	These vectors define a~$(2d-1)$ dimensional face of~$\mathcal P$, a facet of~$\mathcal P$.
	First, note that~$(\bm w,c)$ is facet defining for~$\mathcal Q$.
	Therefore, there exists a family~$\mathcal T=\{\bm q_i\}_{i\in[d]}$ of~$d$ affinely independent vectors that saturate the inequality.
	These vectors lifted with the parameter~$\bm 0$ generate the family~$\mathcal S:=\{\lambda_{\bm 0}(\bm q_i)\}_{i\in[d]}=\{(\bm q_i,\bm 0)\}_{i\in[d]}$ of~$d$ affinely independent vectors in~$\mathcal P$.
	These vectors trivially saturate the lifted inequality: For all~$i\in[d]$, we have
	\begin{align}
		(\bm w,-\bm w) \cdot (\bm q_i,\bm 0) = \bm w \cdot \bm q_i = c
		\,.
	\end{align}
	Now, we need~$d$ such vectors in addition.
	These are obtained by an appropriate lifting of the extremal vectors.
	Towards that, we use the implication which holds for all vectors~$\bm q=(q_i)_{i\in[d]}\in\mathbb R^{d}$
	\begin{align}
		q_{j} = 0
		\implies
		\bm w\cdot \bm q
		=
		(\bm w,-\bm w) \cdot (\lambda_{\bm e_{j}}(\bm q))
		\,;
		\label{eq:imp}
	\end{align}
	if some~$\bm q$ has a zero entry in dimension~$j$, then lifting that vector with parameter~$\bm e_{j}$ yields a vector that produces the same value for the lifted inequality as~$\bm q$ for the original.
	This implication is immediate:
	\begin{align}
		(\bm w,-\bm w)\cdot (\bm q\oplus \bm e_{j}, \bm e_{j})
		=
		\sum_{i\in[d]\setminus\{j\}}
		w_iq_i
		=
		\bm w\cdot \bm q
		\,.
	\end{align}
	This means that whenever~$\bm q\in\ext(\mathcal Q)$ saturates~$(\bm w,c)$ and has a zero entry in dimension~$j$, then~$\lambda_{\bm e_j}(\bm q)$ saturates~$((\bm w,-\bm w),c)$.
	Suppose now there exists a sequence of the vectors in~$\mathcal T$, such that the~$k$th vector has a zero in dimension~$k$, {\it i.e.,} suppose there exists a function~$f:[d]\rightarrow[d]$, such that~$(\bm q_{f(k)})_k = 0$.
	Then, the additional~$d$ vectors we are looking for are easily obtained:
	\begin{align}
		\mathcal S':=
		\{
			\lambda_{\bm e_k}
			(
			\bm q_{f(k)}
			)
		\}_{k\in[d]}
		=
		\{
			(
			\bm q_{f(k)}
			\oplus
			\bm e_k
			,
			\bm e_k
			)
		\}_{k\in[d]}
		\,.
	\end{align}
	These vectors saturate the lifted inequality due to Eq.~\eqref{eq:imp}, and are affinely independent (also with respect to~$\mathcal S$), because each of the vectors contributes to an otherwise untouched dimension.

	What remains to be shown is that such a function~$f$ exists for~$\mathcal T$.
	Towards a contradiction, suppose no such~$f$ exists.
	This means that there exists a dimension~$\ell$ such that~\mbox{$\forall \bm t\in\mathcal T: t_\ell = 1$}.
	Without loss of generality, we take~$\ell$ to be the first dimension, {\it i.e.,}~$\ell=0$, and we have~$\forall \bm v\in\aff(\mathcal T): v_0=1$.
	Since~$\aff(\mathcal T)$ is~$(d-1)$-dimensional and the first dimension is fixed to take value~$1$, the vector~$(1,\bm 0)$ and every vector from~$\{(1,\bm e_k)\}_{k\in[d-1]}$ is in~$\aff(\mathcal T)$.
	Moreover, note that~$(\bm w,c)$ is saturated by any affine combination of the vectors in~$\mathcal T$.
	From this we get that for all~$k\in[d-1]$:
	\begin{align}
		c
		=
		\bm w\cdot (1, \bm 0)
		=
		\bm w\cdot (1, \bm e_k)
		\,,
	\end{align}
	and the inequality~$(\bm w,c)$ must be trivial.
\end{proof}

\section{Bell games}
Consider the following general scenario.
For~$n$-parties~$[n]$, each party~$k\in[n]$ receives an input~$x_k\in\mathcal X_k$, for some input space~$\mathcal X_k$, and {\em immediately\/} produces an output~$a_k\in\mathcal A_k$, for some output space~$\mathcal A_k$.
The realization of~$a_k$ is the physical event~$E_k$.
An input-output behavior of these parties is described by the conditional probability distribution~$p(\underline a| \underline x)$.
We assume that the causal order of the parties is static, {\it i.e.,} the events~$\{E_k\}_{k\in[n]}$ form a partial order~$\preceq_\sigma$.
This means that if~$j\not\preceq_\sigma i$, then the output~$a_i$ is {\em independent\/} of~$x_j$ and~$a_j$.
In general, the partial order among the events may depend on some initial randomness, {\it e.g.,} a coin flip.
This leads to the following definition.\footnote{In fact, this definition is equivalent in asking~$p(\underline a|\underline x)$ to decompose with total orders instead of partial orders.}
\begin{definition}[Partial-order correlations]
	\label{def:pocorr}
	The~$n$-party correlations~$p(\underline a|\underline x)$ are {\em partial-order correlations\/} if and only if
	there exists a probability distribution~$p(\sigma)$ over partial orders~$\preceq_\sigma$ and a family of conditional probability distributions~$\{p_{k,\sigma}(a_k|a_{\prec_\sigma k}, x_{\preceq_\sigma k})\}_{k,\sigma}$, such that
	\begin{align}
		p(\underline a| \underline x)
		=
		\sum_{\sigma}
		p(\sigma)
		\prod_{k\in[n]}
		p_{k,\sigma}(a_k|a_{\prec_\sigma k}, x_{\preceq_\sigma k})
		\,.
	\end{align}
\end{definition}

\subsection{Single-output scenario}
We simplify the general scenario to the minimal case of a single bit of output.
The party who produces this single bit of output, however, is selected by the input.
\begin{definition}[Single-output scenario]
	\label{def:sos}
	For~\mbox{$n\geq 2$} parties~$[n]$, the input space of party~$k\in[n]$ is~\mbox{$\mathcal X_k:=[n]^2_{\neq} \times\mathcal Z_k$},
	where the space~$\mathcal Z_k$ depends on the input to the~$[n]^2_{\neq}$ part.
	Similarly, the output space~$\mathcal A_k$ depends on the input to the~$[n]^2_{\neq}$ part.
	For~$(s,r)\in[n]^2_{\neq}$,
	\begin{align}
		\mathcal Z_k = 
		\begin{cases}
			[2] & \text{if }k=s,\\
			\emptyset & \text{otherwise,}
		\end{cases}
		\quad
		\mathcal A_k = 
		\begin{cases}
			[2] & \text{if }k=r,\\
			\emptyset & \text{otherwise.}
		\end{cases}
	\end{align}
	The input~$(s,r)\in[n]^2_{\neq}$ is {\em shared\/} with all parties.
\end{definition}
Note that the above definition could also be formulated with fixed input spaces and a special symbol to denote ``no input.''
Also, one could define the output space to be the binary space~$[2]$ for {\em all\/} parties, and then marginalize over the non-relevant outputs.
Such a definition, however, would be redundant.
Also, note that the {\em shared input\/} space~$[n]^2_{\neq}$ has an intuitive role:
If the shared input is~$(s,r)$, then party~$s$ has the additional binary input space and party~$r$ has the binary output space; party~$s$ is the ``sender'' and party~$r$ is the ``receiver.''
It is important to note, however, that the parties without output are thought of producing a trivial output.
This renders the above-mentioned events~$\{E_k\}_{k\in[n]}$  well-defined.

In this single-output scenario, the input-output behavior is some conditional probability distribution~$p(a|s,r,x)$ where the pair~$(s,r)$ is an input to every party,~$x$ is the binary input to party~$s$, and~$a$ is the binary output of party~$r$.
Such correlations~$p(a|s,r,x)$, according to Definition~\ref{def:pocorr}, are partial-order correlations whenever there exists a distribution~$p(\sigma)$ over all partial orders, and a family of conditional probability distributions~$\{p_{\sigma}^{\preceq}(a|s,r,x),p_{\sigma}^{\not\preceq}(a|s,r)\}_{\sigma}$, such that they decompose as
\begin{align}
	\begin{split}
		p(a | s,r,x)
		&=
		\sum_{\sigma: s\preceq_\sigma r} p(\sigma) p_{\sigma}^{\preceq}(a | s,r,x)
		\\+
		&\quad\sum_{\sigma: s\not\preceq_\sigma r} p(\sigma) p_{\sigma}^{\not\preceq}(a | s,r)
		\,;
	\end{split}
	\label{eq:correlations}
\end{align}
the output~$a$ of the ``receiver''~$r$ may depend on the bit~$x$ only if the ``sender''~$s$ is before~$r$.
\begin{definition}[Single-output partial-order correlations]
	For~$n\geq 2$ parties, the set~$\mathcal C_n$ is the set of all conditional probability distributions~$p(a|s,r,x)$ that decomposed as in Eq.~\eqref{eq:correlations}.
\end{definition}

\subsection{Geometric representation}
In order to derive the limits on partial-order correlations, we represent the attainable correlations geometrically.
The facet-defining inequalities of the resulting polytope correspond to the respective Bell games (see Figure~\ref{fig:polytope}).
\begin{figure}
	\centering
	\begin{tikzpicture}
		\pgfmathsetmacro{\delta}{0.5}
		\pgfmathsetmacro{\scale}{1.75}
		\path[draw,fill=black!20!white] (0,0) -- (0,1*\scale) -- (0.5*\scale,1.25*\scale) -- (1.25*\scale,0.5*\scale) -- (1*\scale,0) -- cycle;
		\draw[red,thick] (0.5*\scale-\delta*\scale,1.25*\scale+\delta*\scale) -- (1.25*\scale+\delta*\scale,0.5*\scale-\delta*\scale);
		\node[point,label={below right:$\bm p$}] (p) at (1.25*\scale,1*\scale) {};
	\end{tikzpicture}
	\caption{Schematic of a polytope with a facet-defining inequality.
	The shaded region represents the polytope~$\mathcal P_n$ of partial-order correlations. The line indicates a~facet-defining inequality.
	Every point left to or on that line satisfies the inequality.
	In contrast, point~$\bm p$ {\em violates\/} the inequality.
	This violation proves~$\bm p\not\in\mathcal P_n$.}
	\label{fig:polytope}
\end{figure}
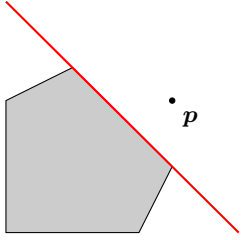
The geometric representation is obtained by expressing the conditional probability distributions from the set~$\mathcal C_n$ as vectors.
Thanks to total probability, the characteristic vector of some~$p\in\mathcal C_n$ is the~$2n(n-1)$-dimensional vector of the probability that~$a=0$ for each input, {\it i.e.,} each~$p\in\mathcal C_n$ is represented by~$\bm \chi(p)$ with
\begin{align}
	\bm\chi: \mathcal C_n&\rightarrow \mathbb R^{2n(n-1)}\\
	p&\mapsto 
	(p(0|s,r,x))_{(s,r,x)\in[n]^2_{\neq}\times[2]}
	\,.
\end{align}
We state and show some immediate facts.
\begin{lemma}[Partial-order-correlations polytope~$\mathcal P_n$]
	\label{lemma:geometricrepresentation}
	For~$n\geq 2$ parties, the set~\mbox{$\mathcal P_n:=\bm\chi(\mathcal C_n)\subseteq\mathbb R^{2n(n-1)}$} is a {\em full-dimensional 0/1 polytope.}
	We call this set the~{\em $n$-party partial-order-correlations polytope.}
\end{lemma}
\begin{proof}
	That~$\mathcal P_n$ is a {\em convex 0/1 polytope\/} follows from the fact that~$\bm\chi$ is linear,~$\mathcal C_n$ is convex, and
	any conditional probability distribution that satisfies Eq.~\eqref{eq:correlations} can be written as a convex combination of deterministic ones.
	Finally, that~$\mathcal P_n$ is {\em full dimensional\/} follows from observing that~$\mathcal P_n$ contains the zero vector and all unit vectors in~$\mathbb R^{2n(n-1)}$.
\end{proof}

Since the parties may arbitrarily relabel the output, the polytope~$\mathcal P_n$ of interest is moreover pairwise centrally symmetric.
\begin{lemma}[Pairwise central symmetry]
	\label{lemma:pcs}
	For~$n\geq 2$, the polytope~$\mathcal P_n$ is pairwise centrally symmetric.
\end{lemma}
\begin{proof}
	Let~$\bm p=(p_{s,r,x})_{s,r,x}\in\ext(\mathcal P_n)$, and partition the vector~$\bm p$ into the two~$n(n-1)$ dimensional vectors~\mbox{$\bm p_b=(p_{s,r,x=b})_{s,r}$} by restricting the input~$x$ to zero or one.
	If for some arbitrary~$(\hat s,\hat r)\in[n]_{\neq}^2$ party~$\hat r$ flips the output bit~$a$, then the resulting vector is as~$\bm p$ but where the entries at the coordinates~$(\hat r,\hat s,0),(\hat r,\hat s,1)$ are flipped.
\end{proof}

\subsection{Projection of~$\mathcal P_n$}
Since the polytope~$\mathcal P_n$ is pairwise centrally symmetric, we know from Theorem \ref{thm:liftingfacets} that we can find its facet-defining inequalities by computing the facets of the projected polytope (see Definition~\ref{def:projection})
\begin{align}
	\mathcal Q_n:=\conv(\pi_{n(n-1)}(\ext(\mathcal P_n)))
	\,.
\end{align}
As it turns out, the extremal vectors of~$\mathcal Q_n$ have a simple form.
This allows us to refer to~$\mathcal Q_n$ as the {\em DAG polytope.}
\begin{lemma}[DAG]
	\label{lemma:dag}
	For~$n\geq 2$, the set of extremal vectors of~$\mathcal Q_n$ equals the set of the adjacency vectors of all DAGs of order~$n$, {\it i.e.,}~$\ext(\mathcal Q_n)=\bm\alpha(\DAG_n)$.
\end{lemma}
\begin{proof}
	Both sets consist of~$(n(n-1))$-dimensional vectors.
	For the~$\subseteq$ inclusion, let~\mbox{$\bm q=(q_{s,r})_{s,r}:=\pi_{n(n-1)}(\bm p)$} for some extremal vector~\mbox{$\bm p=(\bm p_0,\bm p_1)\in\ext(\mathcal P_n)$}.
	The vectors~$\bm p_{b}$ for~\mbox{$b\in\{0,1\}$ are~$(p_{s,r,x=b})_{(s,r)}$}.
	The entry~\mbox{$p_{s,r,x}\in\{0,1\}$} is the probability that party~$r$ outputs~$a=0$ on input~$(s,r)$ to all parties, and on the additional input~$x$ to party~$s$.
	By the projection~$\pi_{n(n-1)}$, the entry~$q_{s,r}$ is~$p_{s,r,0}\oplus p_{s,r,1}$;
	it is zero if the output~$a$ is {\em independent\/} of~$x$, and~$1$ otherwise.
	Therefore,~$\bm q$ is the adjacency vector of the digraph~\mbox{$D=([n],\mathcal A)$}, where~$(s,r)\in\mathcal A$ if and only if the output~$a$ of the ``receiver''~$r$ depends on the input~$x$ of the ``sender''~$s$.
	Since~$\bm p$ is the vector representation~$\bm\chi(p)$ of the partial-order correlations~$p\in\mathcal C_n$, the~$n$ parties that establish these correlations form a partial order.
	Now, for any partial order, it is {\em impossible\/} that the parties influence each other in a cyclic way: 
	The graph~$D$ is a DAG.

	For the reverse direction~$\supseteq$, let~$D=([n],\mathcal A)$ be a DAG.
	In the following, we construct partial-order correlations~$p\in\mathcal C_n$, such that the corresponding vector~$\bm q$ is~$\bm \alpha(D)$.
	First, we find a partial order~$\preceq_\sigma$, such that~$D$ is compatible with~$\preceq_\sigma$ in the following sense:
	if~$(u,v)\in\mathcal A$, then~$u\preceq_\sigma v$.
	Such a partial order can, for instance, be obtained from the transitive closure of the arcs.
	Now, the strategies of the parties are:
	On input~$(u,v)$, the ``receiver''~$v$ outputs~$a=0$ whenever~$(u,v)\not\in\mathcal A$, and~$a=x$ otherwise.
	In the former case, the ``receiver''~$v$ outputs~$a=0$ {\em even if\/}~$u\preceq_\sigma v$---the receiver simply {\em ignores\/} any potential information on~$x$.
	In the latter case, the arc~$(u,v)$ implies that the ``sender''~$u$ is in the causal past of~$v$, and therefore,~$u$ might communicate the value of~$x$ to~$v$.
	These strategies yield the vector~$\bm p=(p_{u,v,x})_{u,v,x}$ with
	\begin{align}
		p_{u,v,x} =
		\begin{cases}
			1 & (u,v)\not\in\mathcal A\,,\\
			1 & (u,v)\in\mathcal A \wedge x=0\,,\\
			0 & (u,v)\in\mathcal A \wedge x=1\,.
		\end{cases}
	\end{align}
	Therefore, the parity~$p_{u,v,0}\oplus p_{u,v,1}$ is one if~$(u,v)\in\mathcal A$, and zero otherwise:~$\bm q=\pi_{n(n-1)}(\bm p)=\bm \alpha(D)$.
\end{proof}

\subsection{The acyclic-subdigraph polytope}
The (unweighted) acyclic-subdigraph problem in combinatorial optimization is the following:
Given a digraph~$D=(\mathcal V,\mathcal A)$, find a DAG subdigraph~$D'$ of~$D$ which maximizes the number of arcs.
The dual of this problem is known as the minimum-feedback-arc-set problem.
Algorithms to these problems are of practical relevance in {\it e.g.,} voting, ranking, and task scheduling~\cite[p.~57--61]{junger1985}.
One can associate to every instance of such a problem a polytope with the solutions as extremal points.
The representation of the polytope in terms of linear inequalities allows for the application of linear-programming techniques.
The polytope for the acyclic-subdigraph problem is as follows:
\begin{definition}[Acyclic-subdigraph polytope~\cite{acyclic1985}]
	Given a digraph~$D=(\mathcal V,\mathcal A)$, the acyclic-subdigraph polytope is
	\begin{align}
		\AC(D)
		:=
		\conv(
		\{\bm \alpha(D') \mid \DAG \ni D'\subseteq D \}
		)
		\,.
	\end{align}
\end{definition}
Thus, the polytope~$\mathcal Q_n$ of our interest is~$\AC(K_n^\text{di})$, and we have
\begin{align}
	\ext(\mathcal Q_n)
	=
	\ext(\AC(K_n^\text{di}))
	=
	\bm\alpha(\DAG_n)
	\,.
\end{align}
Gr\"otschel, J\"unger, and Reinelt~\cite{acyclic1985} derive a series of facet-defining inequalities for~$\AC(D)$, and hence, also for~$\mathcal Q_n$. 
Before we restate these inequalities, we define two classes of digraphs.
Note that the graphs in both classes are bipartite.
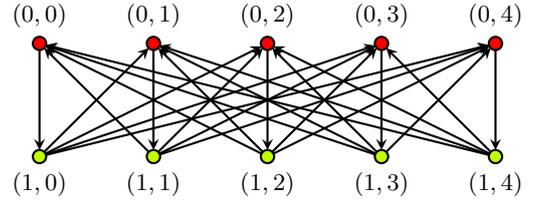
\begin{figure}
	\centering
	\begin{tikzpicture}[every path/.style={-stealth,thick}]
		\def\k{5}
		\def\xscale{1.5}
		\def\yscale{1.5}
		\foreach \i in {1,2,...,\k} {
			\pgfmathtruncatemacro{\id}{\i-1}
			\node[color1,label={above:$(0,\id)$}] (v\i) at (\i*\xscale,\yscale) {};
			\node[color2,label={below:$(1,\id)$}] (w\i) at (\i*\xscale,0) {};
		}
		\foreach \i in {1,2,...,\k} {
			\draw (v\i) -- (w\i);
		}
		\foreach \i in {2,3,...,\k} {
			\draw (w1) -- (v\i);
		}
		\foreach \i in {1,3,4,5} {
			\draw (w2) -- (v\i);
		}
		\foreach \i in {1,2,4,5} {
			\draw (w3) -- (v\i);
		}
		\foreach \i in {1,2,3,5} {
			\draw (w4) -- (v\i);
		}
		\foreach \i in {1,2,3,4} {
			\draw (w5) -- (v\i);
		}
	\end{tikzpicture}
	\caption{The~$5$-fence digraph with its two-coloring.}
	\label{fig:fence}
\end{figure}
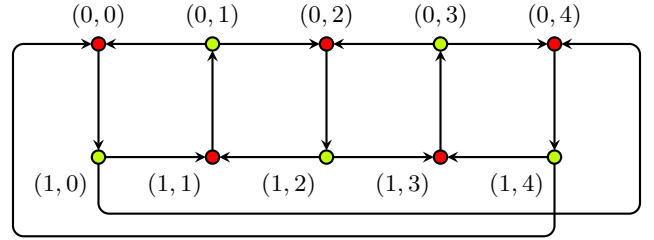
\begin{figure}
	\centering
	\begin{tikzpicture}[every path/.style={-stealth,thick}]
		\def\k{5}
		\def\xscale{1.5}
		\def\yscale{1.5}
		\pgfmathsetmacro{\tang}{90}
		\pgfmathsetmacro{\bang}{260}
		\foreach \i in {1,3,...,\k} {
			\pgfmathtruncatemacro{\id}{\i-1}
			\node[color1,label={\tang:$(0,\id)$}] (v\i) at (\id*\xscale,\yscale) {};
			\node[color2,label={\bang:$(1,\id)$}] (w\i) at (\id*\xscale,0) {};
		}
		\foreach \i in {2,4,...,\k} {
			\pgfmathtruncatemacro{\id}{\i-1}
			\node[color2,label={\tang:$(0,\id)$}] (v\i) at (\id*\xscale,\yscale) {};
			\node[color1,label={\bang:$(1,\id)$}] (w\i) at (\id*\xscale,0) {};
		}
		\foreach \i in {1,3,...,\k} {
			\draw (v\i) -- (w\i);
		}
		\foreach \i in {2,4,...,\k} {
			\draw (w\i) -- (v\i);
		}
		\draw (w1) -- (w2);
		\draw (w3) -- (w4);
		\draw (v2) -- (v1);
		\draw (v4) -- (v3);
		\draw (v2) -- (v3);
		\draw (w3) -- (w2);
		\draw (v4) -- (v5);
		\draw (w5) -- (w4);
		\pgfmathtruncatemacro{\km}{\k-1}
		\draw[rounded corners] (w1) |- ++(\km*\xscale+0.75*\xscale,-0.5*\yscale) |- (v5);
		\draw[rounded corners] (w5) |- ++(-\km*\xscale-0.75*\xscale,-0.7*\yscale) |- (v1);
	\end{tikzpicture}
	\caption{The~$5$-M\"obius digraph with its two-coloring.}
	\label{fig:moebiusX}
\end{figure}
\begin{definition}[$k$-fence and~$k$-M\"obius digraph]
	The {\em~$k$-fence digraph\/} for~$k\geq 2$ (see Figure~\ref{fig:fence}) is the orientation of~$K_{k,k}$ where~$\{(0,i),(1,i)\}$ is oriented as~$((0,i),(1,i))$, and all remaining edges are oriented as~$((1,i),(0,j))$, {\it i.e.,} it has the vertex set~$\{0,1\}\times[k]$, and the arcs
	\begin{align}
		&
		\Big\{
			(
			(0,i)
			,
			(1,i)
			)
			,
			(
			(1,i)
			,
			(0,j)
			)
			\,
			\Big|
			\,
			(i,j)\in[k]^2_{\neq}
		\Big\}
		\,.
	\end{align}

	The {\em~$k$-M\"obius digraph\/} (see Figure~\ref{fig:moebiusX}) is defined for odd~$k\geq 3$ only, and it is the orientation of the~$k$-ladder~\mbox{$L_2\times L_k$} such that~$\{(0,0),(1,0)\}$ is oriented as~$((0,0),(1,0))$, the arcs adjacent to every internal face form a four-cycle,
	and has the ``crossing'' arcs~$\{((1,0),(0,k-1)),((1,k-1),(0,0))\}$ in addition, {\it i.e.,} it has the vertex set~$\{0,1\}\times[k]$ and the arcs
	\begin{align}
		\begin{split}
			&
			\Big\{
				\big(
				(x,i+\gamma y)
				,
				(y, i+\gamma(x\oplus 1))
				\big)
				\,
				\Big|
				\,
				\\
				&\quad
				i\in[k],
				i\text{ odd},
				x,y\in\{0,1\},
				\gamma\in\{\pm 1\}
			\Big\}
			\;
			\cup
			\\
			&
			\Big\{
				((1,0),(0,k-1)),((1,k-1),(0,0))
			\Big\}
			\,.
		\end{split}
	\end{align}
\end{definition}
Note that the~$2$-fence digraph is isomorphic to the~$4$-cycle digraph, and that the~$3$-fence digraph is isomorphic to the~$3$-M\"obius digraph.
Also, the~$k$-M\"obius digraph is an orientation of the~$2k$-M\"obius ladder defined in Ref.~\cite{guy1967}.

For the following theorem, we need the notion of {\em trivial embedding.}
Let~$D=(\mathcal V\subseteq[n],\mathcal A)$ be some digraph.
We can always embed~$D$ in a digraph~$D^{\uparrow n}:=([n],\mathcal A)$ of order~$n$ by extending the set of vertices from~$\mathcal V$ to~$[n]$, and by keeping the set of arcs~$\mathcal A$.
\begin{theorem}[Facet-defining inequalities~\cite{acyclic1985}]
	\label{thm:acyclic1985}
	For~$n\geq 2$, let~$D=(\mathcal V\subseteq [n],\mathcal A)$ be some digraph.

	If~$D$ is isomorphic to the~$k$-cycle digraph~$C_k$, then the {\em~$k$-cycle inequality\/}~$(\bm \alpha(D^{\uparrow n}),k-1)$ is facet-defining for~$\mathcal Q_n$.

	If~$D$ is isomorphic to the~$k$-fence digraph, then the {\em~$k$-fence inequality\/}~$(\bm \alpha(D^{\uparrow n}),k^2-k+1)$ is facet-defining for~$\mathcal Q_n$.

	If~$D$ is isomorphic to the~$k$-M\"obius digraph, then the {\em~$k$-M\"obus inequality\/}~$(\bm \alpha(D^{\uparrow n}),(5k-1)/2)$ is facet-defining for~$\mathcal Q_n$.
\end{theorem}
Many variations of these facet-defining inequalities are known (see Ref.~\cite{goemans1996} for a list).
While Theorem \ref{thm:liftingfacets} holds for all facet-defining inequalities of the polytope~$Q_n$, we only focus on these simple forms.

We briefly discuss these inequalities.
For the digraph~$D=([3],\{(0,1),(1,0),(2,0),(2,1)\})$ (see Figure~\ref{fig:embeddingcycle}),
\begin{figure}
	\centering
	\begin{subfigure}[t]{.2\textwidth}
		\begin{tikzpicture}[every path/.style={-stealth,thick}]
			\node[vertex,label={below:$2$}] (v2) at (0,0) {};
			\node[vertex,label={above left:$0$}] (v0) at (-1,1) {};
			\node[vertex,label={above right:$1$}] (v1) at (+1,1) {};
			\draw[line width=1.5mm,red!40!white] (v0) to[bend left] (v1);
			\draw[line width=1.5mm,red!40!white] (v1) to[bend left] (v0);
			\draw (v2) -- (v0);
			\draw (v2) -- (v1);
			\draw (v0) to[bend left] (v1);
			\draw (v1) to[bend left] (v0);
		\end{tikzpicture}
		\caption{ }
		\label{fig:embeddingcycle}
	\end{subfigure}
  ~
	\begin{subfigure}[t]{.2\textwidth}
		\begin{tikzpicture}[every path/.style={-stealth,thick}]
			\def\xscale{1}
			\def\yscale{1}
			\node[vertex,label={above:$0$}] (v0) at (0,1) {};
			\node[vertex,label={above:$4$}] (w1) at (1,1) {};
			\node[vertex,label={above:$2$}] (v2) at (2,1) {};
			\node[vertex,label={left:$3$}] (w0) at (0,0) {};
			\node[vertex,label={below:$1$}] (v1) at (1,0) {};
			\node[vertex,label={right:$5$}] (w2) at (2,0) {};
			\draw[line width=1.5mm,red!40!white] (v0) -- (w0);
			\draw[line width=1.5mm,red!40!white] (v1) -- (w1);
			\draw[line width=1.5mm,red!40!white] (v2) -- (w2);
			\draw[line width=1.5mm,red!40!white] (w0) -- (v1);
			\draw[line width=1.5mm,red!40!white] (w1) -- (v0);
			\draw[line width=1.5mm,red!40!white] (w1) -- (v2);
			\draw[line width=1.5mm,red!40!white] (w2) -- (v1);
			\draw[line width=1.5mm,red!40!white,rounded corners] (w0) |- ++(2.5,-0.5) |- (v2);
			\draw[line width=1.5mm,red!40!white,rounded corners] (w2) |- ++(-2.5,-0.75) |- (v0);
			\draw[dotted] (v0) -- (w0);
			\draw[dotted] (v1) -- (w1);
			\draw[dotted] (v2) -- (w2);
			\draw (w0) -- (v1);
			\draw (w1) -- (v0);
			\draw (w1) -- (v2);
			\draw (w2) -- (v1);
			\draw[rounded corners] (w0) |- ++(2.5,-0.5) |- (v2);
			\draw[rounded corners] (w2) |- ++(-2.5,-0.75) |- (v0);
		\end{tikzpicture}
		\caption{ }
		\label{fig:embeddingfence}
	\end{subfigure}
	\caption{(a) The digraph~$D$ on top of the~$2$-cycle digraph~$C_2$. (b) A weighted digraph (dotted arcs have weight 1/2) on top of the~$3$-fence digraph.}
	\label{fig:embedding}
\end{figure}
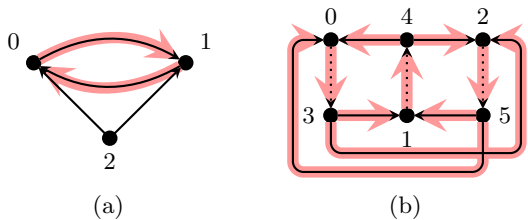
the vector~$\bm \alpha(D)$ is {\em not\/} in the polytope~$\mathcal Q_3$ (obviously,~$D$ is not a DAG), because the~$2$-cycle inequality is violated:
\begin{align}
	\bm\alpha(C_2^{\uparrow 3})\cdot\bm\alpha(D) = 2 > 1
	\,.
\end{align}
The scalar product in the evaluation of the inequality corresponds to counting the number of the arcs simultaneously present in both digraphs.
For another illustration, let~$\bm v$ be the vector that corresponds to the uniform mixture of two digraphs~$(\bm \alpha(D_0) + \bm\alpha(D_1))/2$,
where~$D_0$ is isomorphic to the~$3$-fence digraph, shown in Figure~\ref{fig:embeddingfence}, and~$D_1$ is the same digraph with the vertical arcs~$\{(0,3),(1,4),(2,5)\}$ missing.\footnote{The vector~$\bm v$ can be understood as the adjacency vector of a {\em weighted\/} digraph.}
First, note that~$\bm v$ violates the~$3$-fence inequality:
\begin{align}
	\bm\alpha(D_0)\cdot\bm v = 7.5 > 7
	\,.
\end{align}
However, the vector~$\bm v$ does {\em not\/} violate any~$k$-cycle inequality.
For instance, consider digraph~$\hat C$ with the arcs~$\{(0,3),(3,1),(1,4),(4,0)\}$ and vertices~$\{0,1,\dots,5\}$.
When we evaluate the corresponding~$4$-cycle inequality, we obtain
\begin{align}
	\bm\alpha(\hat C)\cdot\bm v = 3 \leq 3
	\,;
\end{align}
the inequality is satisfied.
This illustrates that the~$k$-cycle inequalities are {\em insufficient\/} to single out the DAG polytope~$\mathcal Q_n$.
For~$0/1$ vectors~$\bm v\in\{0,1\}^{n(n-1)}$ the~$k$-cycle inequalities are---by definition---sufficient to detect whether~$\bm v$ corresponds to the adjacency vector of a DAG or not.
This insufficiency for non-$0/1$ vectors means that the DAG polytope has additional structure that emerges with increasing order (see also Refs.~\cite{acyclic1985,junger1985}).

\subsection{Partial-order inequalities}
We introduce the concept of a digraph game:
\begin{definition}[Digraph game]
	For~$n$ parties~$[n]$ and a digraph~\mbox{$D=(\mathcal V\subseteq[n],\mathcal A)$}, the {\em digraph game~$\Gamma(n,D)$\/} is as follows.
	First, a referee picks at random an arc~$(s,r)$ from~$\mathcal A$, and a bit~$x\in\{0,1\}$.
	Then, the referee announces~$(s,r)$ to all parties~$[n]$, and in addition, distributes~$x$ to party~$s$.
	We say that the~$n$ parties {\em win\/} the game~$\Gamma(n,D)$ whenever the output of party~$r$ equals~$x$, and denote this event by~$\mathcal W(\Gamma(n,D))$.
\end{definition}
We may now combine all previous lemmas with Theorem~\ref{thm:acyclic1985}, and obtain Bell games with which dynamical causal relations among the parties are detected.
\begin{theorem}[Bell games]
	\label{thm:bellgames}
	Consider~$n$ parties~$[n]$ and a digraph~$D=(\mathcal V\subseteq[n],\mathcal A)$.
	If~$(\bm \alpha(D^{\uparrow n}),c)$ is a facet-defining inequality of~$\DAG_n$, then the maximum winning probability of the game~$\Gamma(n,D)$ with partial-order correlations is bounded by
	\begin{align}
		\max_{p\in\mathcal C_n}
		\Pr[\mathcal W(\Gamma(n,D))]
		\leq
		\frac{1}{2}
		+
		\frac{c}{2|\mathcal A|}
		\,.
	\end{align}
	In particular, if~$D$ is isomorphic to the~$k$-cycle digraph, then the bound is~$1-\frac{1}{2k}$,
	if~$D$ is isomorphic to the~$k$-fence digraph, then the bound is~$1-\frac{k-1}{2k^2}$,
	and if~$D$ is isomorphic to the~$k$-M\"obius digraph, then the bound is~$1-\frac{k+1}{12k}$.
	Each of these inequalities is facet defining for the partial-order-correlations polytope~$\mathcal P_n$.
\end{theorem}
\begin{proof}
	The set of single-output correlations among~$n$ partially ordered parties forms a polytope~$\mathcal P_n$ (Lemma~\ref{lemma:geometricrepresentation}) which is pairwise centrally symmetric (Lemma~\ref{lemma:pcs}).
	By applying the natural projection (Definition~\ref{def:projection}), we obtain the polytope~$\mathcal Q_n$ where all extremal vectors are the adjacency vectors of DAGs (Lemma~\ref{lemma:dag}).
	From Theorem~\ref{thm:acyclic1985}, we get some facet-defining inequalities for~$\mathcal Q_n$.
	All these inequalities are non-trivial and non-negative, and the lifting theorem (Theorem~\ref{thm:liftingfacets}) is applicable.
	Therefore, in the respective cases where~$D^{\uparrow n}$ is the trivial embedding of a~$k$-cycle,~$k$-fence, or~$k$-M\"obius digraph, we have as facet-defining inequalities for~$\mathcal P_n$
	the~$k$-cycle inequality~$((\bm\alpha(D^{\uparrow n}),-\bm\alpha(D^{\uparrow n})), k-1)$,
	the~$k$-fence inequality~$((\bm\alpha(D^{\uparrow n}),-\bm\alpha(D^{\uparrow n})), k^2-k+1)$, and
	the~$k$-M\"obius inequality~\mbox{$((\bm\alpha(D^{\uparrow n}),-\bm\alpha(D^{\uparrow n})), (5k-1)/2)$}.

	Now, let~$p\in\mathcal C_n$ be some single-output partial-order correlations, set~$\bm p$ as its characteristic vector, and let~$((\bm\alpha(D^{\uparrow n}),-\bm\alpha(D^{\uparrow n})),c)$ be a facet-defining inequality from above.
	In evaluating the inequality with respect to~$\bm p$, we get
	\begin{align}
		&
		(\bm\alpha(D^{\uparrow n}),-\bm\alpha(D^{\uparrow n}))\cdot\bm p
		\\
		&
		=
		\sum_{(s,r)\in\mathcal A} p(0|s,r,0)
		-
		\sum_{(s,r)\in\mathcal A} p(0|s,r,1)
		\\
		&
		=
		\sum_{(s,r)\in\mathcal A}
		\left(
		p(0|s,r,0)
		-
		(1-p(1|s,r,1))
		\right)
		\\
		&
		=
		\sum_{\substack{x\in\{0,1\}\\(s,r)\in\mathcal A}}
		p(x|s,r,x)
		-
		|\mathcal A|
		\leq c
		\,.
	\end{align}
	By moving~$|\mathcal A|$ to the right side, and by multiplying the inequality with the uniform probability for the referee announcing~$(s,r)$ and~$x$, we get
	\begin{align}
		\Pr[a=x]
		\leq
		\frac{|\mathcal A| + c}{2|\mathcal A|}
		\,.
	\end{align}
	Finally, we obtain the stated bounds by noting that the~$k$-cycle digraph has~$|\mathcal A|=k$, the~$k$-fence digraph has~\mbox{$|\mathcal A|=k^2$}, and the~$k$-M\"obius digraph has~$|\mathcal A|=3k$.
\end{proof}

Note that for increasing~$k$, the bounds for the~$k$-cycle and the~$k$-fence games approach one; they may be won with partial-order correlations.
In contrast, the winning chance of the~$k$-M\"obius game under the same conditions never exceeds~$11/12$.
This elevates the~$k$-M\"obius game as a preferred and robust test for dynamical causal order, and shows that dynamical causal order might be unlimited:
That feature does not vanish for large numbers of events.
The same quality has been observed for nonlocal correlations~\cite{collins2002,seevinck2002} and non-causal correlations~\cite{baumeler2020}.

\subsection{Partial-order and causal inequalities}
Causal inequalities, similar to partial-order inequalities, limit the correlations among~$n$ parties under the assumption of a global, {\em possibly dynamical,\/} causal order.
Here, in contrast with partial-order inequalities, a~party may influence the causal relations of the parties in its causal future.
Clearly, partial-order correlations are a~subset of the causal ones.
Causal inequalities have been extensively studied in the context of indefinite causal order (see, {\it e.g.,} Refs.~\cite{oreshkov2012,baumeler2014,araujo2015,branciard2015,alastaircorr2016,oreshkov2016,baumeler2020}).
While we are mainly concerned about partial-order inequalities, we nevertheless introduce causal correlations and the respective inequalities.
The reason is that they coincide with the former for an even simpler scenario than studied above.
Moreover, we recover a Bell game that was studied in Ref.~\cite{tselentis2022admissible}.

First, we introduce causal correlations in the general setting.
\begin{definition}[Causal correlations]
	\label{def:causalcorrelations}
	The~$n$-party correlations~$p(\underline a|\underline x)$, where~$x_k$ is the input to party~$k$, and~$a_k$ is the output from party~$k$, are {\em causal\/} if and only if they decomposed as
	\begin{align}
		p(\underline a|\underline x)
		=
		\sum_{i\in[n]}
		p(i)
		p(a_i|x_i)
		p_{a_i,x_i}^i(a_{\setminus \{i\}}|x_{\setminus \{i\}})
		\,,
	\end{align}
	where~$p(i)$ is a probability distribution over~$[n]$, and~$p_{a_i,x_i}^i(a_{\setminus \{i\}}|x_{\setminus \{i\}})$ are~$(n-1)$-party causal correlations.
\end{definition}
In this recursive definition, party~$i$ acts ``first'' (hence, its output may only depend on its input).
The remaining parties have access to~$i,a_i,x_i$, but again, there exists a party that acts ``first.''
The selection of that party may depend on~$i$,~$a_i$, and~$x_i$.

In the single-output scenario (Definition~\ref{def:sos}) studied above, the input~$(s,r)$ specifies the output-providing party~$r$ and the party~$s$ who gets the additional input~$x$.
Instead, one can neglect~$s$, and provide~$x$ to {\em every\/} party, except---to make it non-trivial---to party~$r$.
\begin{definition}[All-to-one scenario]
	For~$n\geq 2$ parties~$[n]$, the input space of party~$k\in[n]$ is~\mbox{$\mathcal X_k:=[n] \times\mathcal Z_k$},
	where the space~$\mathcal Z_k$ depends on the input to the~$[n]$ part.
	Similarly, the output space~$\mathcal A_k$ depends on the input to the~$[n]$ part.
	For~$r\in\mathcal [n]$,
	\begin{align}
		\mathcal Z_k = 
		\begin{cases}
			[2] & \text{if }k\neq r,\\
			\emptyset & \text{otherwise,}
		\end{cases}
		\quad
		\mathcal A_k = 
		\begin{cases}
			[2] & \text{if }k=r,\\
			\emptyset & \text{otherwise.}
		\end{cases}
	\end{align}
	The input to the~$[n]$ part is {\em shared\/} with all parties, and the input to the~$\mathcal Z_k$ part is {\em shared\/} with all parties but~$r$.
\end{definition}
As above, one can define the set of partial-order correlations~$\mathcal C^\text{all-to-1}_n$ for this simplified scenario and study the polytope~$\mathcal P^\text{all-to-1}_n:=\bm\chi'(\mathcal C^\text{all-to-1}_n)\in\mathbb R^{2n}$,
and likewise for causal correlations.
Here, the characteristic vector of some conditional probability distribution~$p\in\mathcal C^\text{all-to-1}_n$ is obtained from
\begin{align}
	\bm\chi': \mathcal C^\text{all-to-1}_n&\rightarrow \mathbb R^{2n}\\
	p&\mapsto 
	(p(0|r,x))_{(r,x)\in[n]\times[2]}
	\,.
\end{align}
Note that in both cases, partial-order and causal, the output~$a$ of party~$r$ {\em must be\/} independent of~$x$ only if~$r$ is first in the partial, respectively causal order.
Hence, we arrive at the following statement, which is trivial to prove.
\begin{lemma}[All-to-one partial-order and causal correlations]
	For~$n$ parties~$[n]$, the correlations~$p(a|r,x)$, where~$a$ is the output of party~$r$,~$r$ is the input to every party, and~$x$ is an input to the parties~$[n]\setminus\{r\}$, are partial-order correlations if and only if
	there exists a distribution~$p(k)$ over all parties,
	and a pair of conditional probability distributions~$(p^{\not\preceq}(a|r),p^\preceq(a|r,x))$, such that
	\begin{align}
		p(a|r,x)
		=
		p(r)
		p^{\not\preceq}(a|r)
		+
		(1-p(r))
		p^{\preceq}(a|r,x)
		\,.
	\end{align}
	These correlations are causal correlations if and only if the same decomposition exists.
\end{lemma}

We directly present our result.
In the proof, we again exploit the property of pairwise central symmetry.
The projected polytope, as shown in the proof, turns out to be the {\em faulty hypercube.}
\begin{theorem}[All-to-one Bell game]
	\label{thm:bellgamesimplified}
	For~$n\geq 2$ parties~$[n]$, a referee picks uniformly at random a ``receiver''~$r\in[n]$, and a bit~$x\in\{0,1\}$,
	and distributes~$r$ to every party, and~$x$ to every party but~$r$.
	The inequality
	\begin{align}
		\Pr[a=x] \leq 1-\frac{1}{2n}
	\end{align}
	is facet-defining for~$\mathcal P^\text{\normalfont{all-to-1}}_n$.
\end{theorem}
\begin{proof}
	The set~$\mathcal P^\text{\normalfont{all-to-1}}_n$ is a full-dimensional convex 0/1 polytope in~$\mathbb R^{2n}$, where a vector~$\bm p=(p_{r,x})_{(r,x)\in[n]\times\{0,1\}}$ is a list of the probability that party~$r$ outputs~\mbox{$a=0$} on input~$r$ to all parties and input~$x$ to the parties~$[n]\setminus\{r\}$.
	This polytope is pairwise centrally symmetric, where~\mbox{$p_0=(p_{r,0})_{r\in[n]}$}, and~\mbox{$p_1=(p_{r,1})_{r\in[n]}$}, for the same reason as~$\mathcal P_n$ is: Each party might relabel the output.
	We apply the projection map (Definition~\ref{def:projection}) to obtain the polytope~$\mathcal Q^{\text{\normalfont{all-to-1}}}_n:=\conv(\pi_n(\ext(\mathcal P^{\text{\normalfont{all-to-1}}}_n)))\subseteq\mathbb R^{n}$.
	As it turns out, this polytope is the faulty~$n$-cube
	\begin{align}
		\ext(\mathcal Q^{\text{\normalfont{all-to-1}}}_n )
		=
		\{0,1\}^n
		\setminus
		\bm 1
		\,.
	\end{align}
	To see this, let~$\bm q=(q_r)_{r\in[n]}=\pi_n(\bm p)$ be an element of the left-hand side.
	The entries of this vector are~\mbox{$q_r=p_{r,0}\oplus p_{r,1}$}.
	The entry~$q_r$ is zero if the output~$a$ of party~$r$ does {\em not\/} depend on~$x$, and is one otherwise.
	Since~$\bm q$ arises from a partial ordering of the parties, we have that there always must exist a party~$r^-$ that is minimal with respect to that partial order.
	The entry~$q_{r^-}$ must be zero, and therefore~$\bm q\neq\bm 1$.
	For the converse, let~$\bm v=(v_r)_{r\in[n]}$ be a vertex of the faulty~$n$-cube, and suppose that for~$r_0\in[n]$ we have~$v_{r_0}=0$.
	Now, take the partial order~$\preceq_\sigma$ where~$r_0$ is minimal, {\it i.e.,} for all~\mbox{$r'\in[n]\setminus\{r_0\}$}, the relation~$r_0\preceq_\sigma r'$ holds.
	From this, we can construct the following strategy.
	If the selected party~$r$ to make the guess is~$r_0$, then party~$r_0$ outputs~$a=0$.
	In the alternative case, where~$r\neq r_0$, party~$r_0$ receives~$x$ from the referee and forwards~$x$ to party~$r$.
	Finally, party~$r$ outputs~$a=q_r x$.

	The faulty hypercube has a single non-trivial facet, which is defined by the inequality~$(\bm 1,n-1)$.
	This facet-defining inequality is non-negative, and therefore, we apply our facet-lifting theorem, and obtain the facet-defining inequality~$((\bm 1,-\bm 1),n-1)$ for the polytope~$\mathcal P^\text{\normalfont{all-to-1}}_n$.
	At last, this inequality
	\begin{align}
		(\bm 1,-\bm 1) \cdot (\bm p_0,\bm p_1)
		&
		=
		\sum_{r\in[n]} p(0|r,0) 
		-
		\sum_{r\in[n]} p(0|r,1) 
		\\
		&=
		\sum_{\substack{x\in\{0,1\}\\r\in[n]}} p(x|r,x) - n
		\\
		&\leq n-1
	\end{align}
	is turned into the Bell game as stated.
\end{proof}

\section{Causal models}
We present strategies with which the presented games are won.
In contrast to the device-independent approach that we followed until now, we give a physical description of the parties, how they are interlinked, and their actions.
This description is given in terms of causal models.
We start with a brief introduction to this framework.
The interested reader may consult Refs.~\mbox{\cite{barrett2021nature,tselentis2022admissible}} for details.

\subsection{Framework}
Consider a setup with~$n$ parties~$[n]$.
Each party~$k\in[n]$ is composed out of a past boundary and a future boundary (see Figure~\ref{fig:processparty}).
Party~$k$ receives a physical system on the past boundary.
The physical system is in a state from the set~$\mathcal I_k$.
After receiving that system, party~$k$ carries out an experiment on that system, and releases a system to the future boundary.
The released system is in a state from the set~$\mathcal O_k$.
In a classical-deterministic world, the experiment carried out by party~$k$ is a function~$\mu_k:\mathcal I_k\rightarrow\mathcal O_k$.
Party~$k$, additionally, may choose a specific experiment based on some experimental setting~$x_k\in\mathcal X_k$.
The experiment may also produce some experimental result~$a_k\in\mathcal A_k$.
Thus, in this general setting, the experiment of party~$k$ is some function~$\mu_k:\mathcal X_k\times\mathcal I_k\rightarrow\mathcal A_k\times\mathcal O_k$.
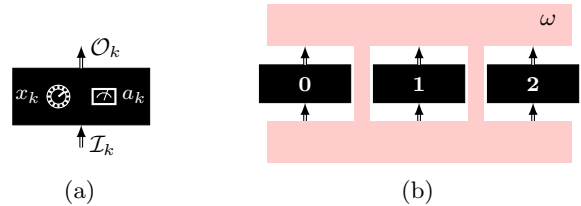
\begin{figure}
	\centering
	\begin{subfigure}[t]{.2\textwidth}
		\begin{tikzpicture}
			\def\scale{1.5}
			\node[box,minimum width=\scale*1.2cm,minimum height=\scale*0.5cm] (A) at (0,0) {};
			\coordinate (knob) at ($ (A.center) + (180:.3) $);
			\node at ($ (knob) + (180:0.4) $) {{\color{white}\footnotesize $x_k$}};
			\draw[white] (knob) circle [radius=0.1];
			\draw[white] (knob) circle [radius=0.15];
			\foreach \x in {1,...,12}
			{
				\def\angle{\x*360/12}
				\draw[white,thick] ($ (knob) + (\angle:0.1) $) -- +(\angle:0.05);
			}
			\def\angle{39}
			\draw[white,thick] (knob.center) -- +(\angle:0.1);
			\coordinate (meter) at ($ (A.center) + (0:.3) $);
			\node at ($ (meter) + (0:0.4) $) {{\color{white}\footnotesize $a_k$}};
			\draw[white,thick] ($ (meter) - (0.15,0.1) $) rectangle ($ (meter) + (0.15,0.1) $);
			\draw[white] ($ (meter) - (0,0.05) $) ++(20:0.1) arc (20:160:0.1);
			\def\angle{70}
			\draw[white] ($ (meter) - (0,0.05) $) -- ++(\angle:0.15);
			\draw[double,-latex] (0,0.5*\scale*0.5cm) -- ++(0,.3) node[pos=.9,right] {$\mathcal O_k$};
			\draw[double,latex-] (0,-0.5*\scale*0.5cm) -- ++(0,-.3) node[pos=.9,right] {$\mathcal I_k$};
		\end{tikzpicture}
		\caption{ }
		\label{fig:processparty}
	\end{subfigure}
  ~
	\begin{subfigure}[t]{.26\textwidth}
		\begin{tikzpicture}
			\def\ah{.25}
			\foreach \x in {0,1,2} {
				\node[box] (A\x) at (1.5*\x,0) {\color{white}\footnotesize\textbf{\x}};
				\draw[double,-latex] (A\x.north) -- ++(0,\ah);
				\draw[double,latex-] (A\x.south) -- ++(0,-\ah);
			}
			\path[fill=red!20!white] ($ (A0.north) + (-.5,\ah) $) rectangle ($ (A2.north) + (.5,.8) $);
			\path[fill=red!20!white] ($ (A0.south) + (-.5,-\ah) $) rectangle ($ (A2.south) + (.5,-.8) $);
			\path[fill=red!20!white] ($ (A0.south) + (.65,-\ah) $) rectangle ($ (A1.north) + (-.65,\ah) $);
			\path[fill=red!20!white] ($ (A1.south) + (.65,-\ah) $) rectangle ($ (A2.north) + (-.65,\ah) $);
			\node at ($ (A2.north) + (.2,.6) $) {$\omega$};
		\end{tikzpicture}
		\caption{ }
		\label{fig:environment}
	\end{subfigure}
	\caption{(a) Party~$k$ receives a system, performs an experiment~$\mu_k$ with experimental setting~$x_k$, observes the experimental result~$a_k$, and releases a system back to the environment. (b) Three parties interlinked by the environment (some process~$\omega$).}
	\label{fig:framework}
\end{figure}

The~$n$ parties are interlinked through the environment (see Figure~\ref{fig:environment}).
The environment takes the physical systems on the future boundaries of the parties, and provides systems to their past boundaries.
Thus, the environment is a function~$\omega:\underline{\mathcal O}\rightarrow\underline{\mathcal I}$.
If we do not request the parties to be situated at some specific locations, but instead assume that each party carries out their experiment exactly once, and that the parties may only communicate through the environment, we arrive at the following description of a process.
\begin{theorem}[Process~\cite{baumeler2016}]
	The function~$\omega:\underline{\mathcal O}\rightarrow\underline{\mathcal I}$ 
	is a {\em classical-deterministic process\/} if and only if
	\begin{align}
		\begin{split}
			\forall
			(\mu_k)_{k\in[n]}
			&
			\in
			(
			\{
			\mathcal I_k\rightarrow\mathcal O_k
			\}
			)_{k\in[n]}
			,\,
			\exists! \underline i\in\underline{\mathcal I}:\\
			\underline i &= \omega(\underline\mu(\underline i))
			\,,
		\end{split}
		\label{eq:fp}
	\end{align}
	where~$\exists!$ is the uniqueness quantifier.
\end{theorem}
This theorem follows from the process-matrix framework~\cite{oreshkov2012}, and states that irrespective of the experiment carried out by the parties, the state of the input system is well-defined.
For a process~$\omega$ and a choice of experiments, the observed statistics are
\begin{align}
	p(\underline a|\underline x)
	=
	\sum_{\underline i,\underline o}
	[\omega(\underline o)=\underline i]
	[(\underline a,\underline o) = \underline\mu(\underline x,\underline i)]
	\,,
\end{align}
where we use~$[i=j]$ for the Kronecker delta~$\delta_{i,j}$.
Satisfying the law of total probability, this expression returns one for experiments without settings and results.

This notion of process is neatly combined with the notion of causal models.
A causal model is a pair:
The causal structure, which is a digraph~\mbox{$M=(\mathcal W,\mathcal S)$} with parties~$\mathcal W$,
and the model parameters, which are a family of functions~\mbox{$\{\omega_k:\mathcal O_{\Pa_M(k)}\rightarrow\mathcal I_k\}_{k\in[n]}$}.
Here,~\mbox{$\Pa_M(k):=\{p\mid(p,k)\in\mathcal S\}\subseteq\mathcal W$} denotes the set of parents of the vertex~$k\in\mathcal W$ with respect to the digraph~$M$.
By combining processes with causal models, we arrive at the following, where the connections among the parties are {\em faithfully represented\/} by the causal structure and where the model parameters are {\em consistent\/} in the sense that the laws of probability are conserved.
\begin{definition}[Faithful and consistent causal model~\cite{tselentis2022admissible}]
	The causal structure~$M$ and the model parameters~\mbox{$\{\omega_k:\mathcal O_{\Pa_M(k)}\rightarrow\mathcal I_k\}_{k\in[n]}$}
	form a {\em faithful and consistent causal model\/} if and only if
	for each vertex~$k\in\mathcal W$, the model parameter~$\omega_k$ depends on every argument, {\it i.e.,} 
	\begin{align}
		\forall p\in\Pa_M(k),\,
		&
		\exists o\in\mathcal O_{\Pa_M(k)\setminus\{p\}},\,
		o_0,o_1\in\mathcal O_p:\\
		&\omega_k(o,o_0) \neq \omega_k(o,o_1)
		\,,
	\end{align}
	and the function~$\omega=(\omega_k)_{k\in[n]}$ is a classical-deterministic process.
\end{definition}

Given a~causal structure~$M$, we build a causal model by defining some natural model parameters, the {\em veto model parameters,} where we use~\mbox{$\Ch_M(k):=\{v\mid(k,v)\in\mathcal S\}\subseteq\mathcal W$} to denote the set of children of the vertex~$k\in\mathcal W$ with respect to the digraph~$M$.
\begin{definition}[Veto model parameters \cite{tselentis2022admissible}]
	The {\em veto  model parameters\/} for the causal structure~$M=(\mathcal W,\mathcal S)$ are
	\begin{align}
		\mathcal I_k := \{0,1\}\,,&\qquad
		\mathcal O_k := \Ch_M(k)\cup\{\bot\}\\
		\omega_k:\mathcal O_{\Pa_M(k)} &\rightarrow \mathcal I_k\\
		(o_\ell)_{\ell\in\Pa_M(k)} &\mapsto  \prod_{\ell\in\Pa_M(k)} [k=o_\ell]
		\,.
	\end{align}
\end{definition}
These model parameters implement the following functionality, which justifies the name.
Each party~$k$ may specify one of its children on its future boundary, or ``nobody'' expressed with the bottom symbol~$\bot$.
If all parents of some party~$k$ specified ``$k$,'' then party~$k$ receives a one on the past boundary, and a zero otherwise.

In the following theorem, we combine some known results on causal models concerning a specific class of digraphs.
On the one hand, when a causal structure from that class is amended with the veto model parameters, the resulting causal model is always {\em faithful and consistent.}
On the other hand, these causal models always give rise to causal correlations only.
\begin{theorem}[Consistency and causal correlations~\cite{tselentis2022admissible}]
	\label{thm:admissible}
	If~$M=(\mathcal W,\mathcal S)$ is a {\em chordless siblings-on-cycles digraph,}
	{\it i.e.,} for each directed cycle~$C=(s_0,s_1,\dots)$ in~$M$ that traverses the vertices~$\mathcal W_C\subseteq\mathcal W$,
	(chordless) the arc set~$\mathcal S$ contains no chord~$(u,v)\in\mathcal W_C^2\setminus C$,
	and (siblings-on-cycles) the arc set~$\mathcal S$ contains at least two arcs~$(p,u_0),(p,u_1)$ with~$u_0,u_1\in\mathcal W_C$,
	then the causal model with causal structure~$M$ and the veto model parameters  
	is {\em a consistent and faithful causal model,}
	and for all experiments~$\{\mu_k\}_{k\in\mathcal W}$, the correlations~$p(\underline a|\underline x)$ are {\em causal.}
\end{theorem}

\subsection{Illustration: The classical switch}
The present framework and this theorem are illustrated by the classical switch.
Consider the~$3$-party scenario, and let the causal structure~$M$ be the digraph~$D$ of Figure~\ref{fig:embeddingcycle}.
In this case, the veto model parameters are
\begin{align}
	\omega_0: \{0,\bot\} \times \{0,1,\bot\}&\rightarrow \{0,1\}\\
	(o_1,o_2) &\mapsto [0=o_1][0=o_2]\\
	\omega_1: \{1,\bot\} \times \{0,1,\bot\}&\rightarrow \{0,1\}\\
	(o_0,o_2) &\mapsto [1=o_1][1=o_2]
	\,.
\end{align}
Since party~$2$ has no parents, the model parameter~$\omega_2$ is simply the constant~$1$.
If party~$2$ implements an experiment such that the system on the future boundary of party~$2$ is in the state~$0$, then the above functions---partially evaluated---become
\begin{align}
	\omega_0(o_1, o_2=0) &= [0=o_1]\\
	\omega_1(o_0,o_2=0) &= 0
	\,;
\end{align}
Party~$0$ may now receive a signal from party~$1$.
If, however, the experiment of party~$2$ produces the state~$1$, then
\begin{align}
	\omega_0(o_1, o_2=0) &= 0\\
	\omega_1(o_0,o_2=0) &= [1=o_0]
	\,,
\end{align}
and party~$1$ may receive a signal from party~$0$---just as in the quantum switch ({\it cf.} Figure~\ref{fig:switch}).
Now, note that~$M$ is a chordless siblings-on-cycles graph.
The only directed cycle is~$( (0,1),(1,0))$.
This cycle is chordless, and there exists~$p=2$ such that~$(p,0),(p,1)$ are arcs.
The above theorem thus tells us that, no matter what experiments are carried out by the parties~$\{0,1,2\}$, only causal correlations (see Definition~\ref{def:causalcorrelations}) are attainable.
This example also illustrates the influence of the common parent~$p=2$:
No matter what experiment~$p$ carries out, the information flow around the directed cycle is effectively interrupted.
In fact, this is what ensures consistency.
Were the information flow along the cycle not interrupted, then a party may influence its own past~\cite{baumeler2020qpl}, which yields a disagreement with Eq.~\eqref{eq:fp}.

\subsection{Game-winning strategies}
We present causal models and experiments with which the digraph Bell games (Theorem~\ref{thm:bellgames}) are won deterministically.
The respective causal structures are particularly simple.
\begin{figure}
	\centering
	\def\rotate{45}
	\def\delta{20}
	\def\radius{1.25}
	\pgfmathsetmacro{\labelradius}{\radius*1.3}
	\def\dd{5}
	\pgfmathsetmacro{\redlast}{\rotate+\delta*0.5}
	\pgfmathsetmacro{\blacklast}{180+\redlast}
	\pgfmathsetmacro{\blackfirst}{\rotate-\delta*0.5}
	\pgfmathsetmacro{\redfirst}{180+\blackfirst}
	\begin{subfigure}[t]{.23\textwidth}
		\centering
		\begin{tikzpicture}[every path/.style={-stealth,thick}]
			\node[vertex] (rlast) at (\redlast:\radius) {};
			\foreach \x in {0,1,2} {
				\pgfmathsetmacro{\a}{\blackfirst-\x*\delta}
				\node[vertex] (b\x) at (\a:\radius) {};
				\pgfmathsetmacro{\a}{\blackfirst-(\x-1)*\delta}
				\draw (\a-\dd:\radius) arc(\a-\dd:\a-\delta+\dd:\radius);
			}
			\def\x{0}
			\pgfmathsetmacro{\a}{\blackfirst-\x*\delta}
			\node (labelb\x) at (\a:\labelradius) {$u_0$};
			\def\x{1}
			\pgfmathsetmacro{\a}{\blackfirst-\x*\delta}
			\node (labelb\x) at (\a:\labelradius) {$y_0$};
			\pgfmathsetmacro{\a}{\blackfirst-2*\delta}
			\draw[dashed] (\a-\dd:\radius) arc(360+\a-\dd:\redlast+\dd:\radius);
			\node[vertex,label=\rotate+90:{$p$}] (p) at (0,0) {};
			\draw (p) -- (b0);
			\draw (p) -- (b1);
		\end{tikzpicture}
		\caption{ }
		\label{fig:winningcycle}
	\end{subfigure}
	~
	\begin{subfigure}[t]{.23\textwidth}
		\centering
		\begin{tikzpicture}[every path/.style={-stealth,thick}]
			\node[color1] (blast) at (\blacklast:\radius) {};
			\node (labelblast) at (\blacklast:\labelradius) {$u_{k-1}$};
			\node[color2] (rlast) at (\redlast:\radius) {};
			\node (labelrlast) at (\redlast:\labelradius) {$y_{k-1}$};
			\foreach \x in {0,1,2} {
				\pgfmathsetmacro{\a}{\blackfirst-\x*\delta}
				\node[color1] (b\x) at (\a:\radius) {};
				\node (labelb\x) at (\a:\labelradius) {$u_{\x}$};
				\node[color2] (r\x) at (\a+180:\radius) {};
				\node (labelr\x) at (\a+180:\labelradius) {$y_{\x}$};
				\pgfmathsetmacro{\a}{\blackfirst-(\x-1)*\delta}
				\draw (\a-\dd:\radius) arc(\a-\dd:\a-\delta+\dd:\radius);
				\draw (180+\a-\dd:\radius) arc(180+\a-\dd:180+\a-\delta+\dd:\radius);
			}
			\pgfmathsetmacro{\a}{\blackfirst-2*\delta}
			\draw[dashed] (\a-\dd:\radius) arc(360+\a-\dd:\blacklast+\dd:\radius);
			\pgfmathsetmacro{\a}{\redfirst-2*\delta}
			\draw[dashed] (\a-\dd:\radius) arc(\a-\dd:\redlast+\dd:\radius);
			\node[vertex,label=\rotate+90:{$p$}] (p) at (0,0) {};
			\draw (p) -- (b0);
			\draw (p) -- (r0);
		\end{tikzpicture}
		\caption{ }
		\label{fig:winningbipartite}
	\end{subfigure}
	\caption{Causal structure to win (a) the~$k$-cycle game for~$k<n$, and (b) the~$k$-fence and the~$k$-M\"obius game for~$2k<n$. Remaining isolated vertices are omitted.}
	\label{fig:winning}
\end{figure}
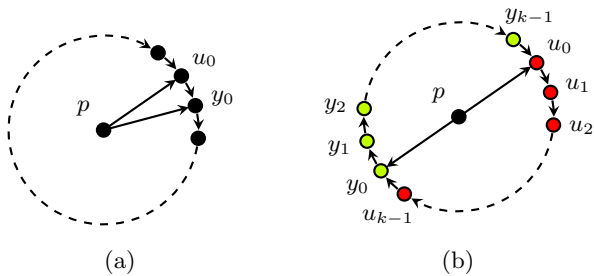
\begin{definition}[Game-winning causal models]%
  \label{def:gwcm}
	The {\em game-winning causal model\/} consists of the following causal structure with the veto model parameters.
	For the~$k$-cycle game~\mbox{$\Gamma(n,D)$} with~\mbox{$D=(\mathcal V\subseteq[n],\mathcal A)$} and~\mbox{$k<n$},
	the causal structure is the digraph~$([n],\mathcal A\cup\{(p,u_0),(p,y_0)\})$,
	where~$p\in[n]\setminus\mathcal V$ and~$(u_0,y_0)\in\mathcal A$ (see Figure~\ref{fig:winningcycle}).
	For the~$k$-fence game~$\Gamma(n,D)$ as well as for the~$k$-M\"obius game with~\mbox{$D=(\mathcal V\subseteq[n],\mathcal A)$} and~\mbox{$2k<n$},
	the causal structure is the digraph with vertices~$[n]$ and arcs
	\begin{align}
		\begin{split}
			&
			\{(u_i,u_{i+1}),(y_i,y_{i+1}) \mid i\in[k-1]\}
			\\
			&
			\cup
			\{(u_{k-1},y_0),(y_{k-1},u_0),(p,u_0),(p,y_0)\}
			\,,
		\end{split}
	\end{align}
	where~$\mathcal U=\{u_i\}_{i\in[k]}$,~$\mathcal Y=\{y_i\}_{i\in[k]}$ correspond to the two-coloring bipartition of~$D$, and~$p\in[n]\setminus\mathcal V$ (see Figure~\ref{fig:winningbipartite}).
\end{definition}
Thanks to Theorem~\ref{thm:admissible}, the game-winning causal model is consistent and faithful.
Next, we specify the experiments to be carried out by the parties.
\begin{definition}[Game-winning experiments]
	For the~$k$-cycle,~$k$-fence, or~$k$-M\"obius game~$\Gamma(n,D=(\mathcal V,\mathcal A))$ and the respective game-winning causal model with the causal structure~$M=([n],\mathcal S)$, the game-winning experiments are the following.
	Whenever the referee announces~$(s,r)\in\mathcal A$ to all parties~$[n]$ and~\mbox{$x\in\{0,1\}$ to party~$s$}, then party~$s$ implements
	\begin{align}
		\mu_s:\mathcal A\times\{0,1\}\times\mathcal I_s &\rightarrow\mathcal O_s\\
		(s,r,x=0,i_s) &\mapsto \bot\\
		(s,r,x=1,i_s) &\mapsto c
		\,,
	\end{align}
	where~$\Ch_M(s)=\{c\}$, {\it i.e.,} party~$s$ ``votes'' for its unique child (with respect to the causal structure ~$M=([n],\mathcal S)$), if and only if~$x=1$,
	party~$r$ (who produces the output~$a$) implements
	\begin{align}
		\mu_{r}:\mathcal A\times\mathcal I_{r} &\rightarrow\{0,1\}\times\mathcal O_{r}\\
		(s,r,i_{r}) &\mapsto (i_{r},\bot)
		\,,
	\end{align}
	{\it i.e.,} party~$r$ produces~$a$ according to the state on the past boundary,
	party~$p$ (who has a controlling role) implements
	\begin{align}
		\mu_{p}:\mathcal A\times\mathcal I_{p} &\rightarrow\times\mathcal O_{p}\\
		(s,r,i_{p}) &\mapsto u_0
	\end{align}
	whenever the directed path from~$s$ to~$r$ in the causal structure~$M$ traverses or ends in~$u_0$, and
	\begin{align}
		(s,r,i_{p}) &\mapsto y_0
	\end{align}
	otherwise,
	and each other party~$k\in\mathcal V\setminus\{s,r,p\}$ implement
	\begin{align}
		\mu_{k}:\mathcal A\times\mathcal I_{k} &\rightarrow\times\mathcal O_{k}\\
		(s,r,i_k=0) &\mapsto \bot\\
		(s,r,i_k=1) &\mapsto c
		\,,
	\end{align}
	where~$\Ch_M(k)=\{c\}$, {\it i.e.,} party~$k$ forwards the signal.
\end{definition}

\begin{theorem}[Tame game-winning causal models]
	\label{thm:winning}
	Let~\mbox{$\Gamma(n,D)$} be a~$k$-cycle game with~\mbox{$k<n$}, a~$k$-fence game with~\mbox{$2k<n$}, or a~$k$-M\"obius game with~\mbox{$2k<n$}.
	If the parties implement the game-winning experiments on the respective game-winning causal model,
	then the parties win the game~$\Gamma(n,D)$ with certainty, {\it i.e.,}~$\Pr[\mathcal W(\Gamma(n,D))]=1$, and yet for all experiments the causal models yields causal correlations only.
\end{theorem}
\begin{proof}
	The latter part follows from Theorem~\ref{thm:admissible}.
	For the former part, suppose the referee announces~$(s,r)\in\mathcal A$ to all parties~$[n]$, and~$x\in\{0,1\}$ to party~$s$, and let~$\pi$ be the directed path from~$s$ to~$r$ on the game-winning causal structure~$M$.
	First, note that~$\pi$ is unique, and in the case of the~$k$-cycle game,~$\pi$ consists of the single arc~$(s,r)$.
	In the case of the~$k$-fence or~$k$-M\"obius game,~$s$ has a different color than~$r$.
	Thus, in all cases, the directed path~$\pi$ traverses or ends in at most one of~$u_0$ and~$y_0$.
	Thanks to the experiment of party~$p$, each party~$k$ on the path~$\pi$, and~$r$, may receive a signal from~$\Pa_M(k)\setminus\{p\}$.
	Now, party~$s$ votes for its unique child~$c$ only if~$x=1$.
	If~$c=r$, then according to~$\mu_c$, party~$r$ outputs~$a=x$.
	Else, party~$c$ forwards the signal to its unique child, {\it etc.}
\end{proof}

\section{Discussion: Relativistic setting}
We exploit the digraph Bell games presented in Theorem~\ref{thm:bellgames} to operationally distinguish between special and general relativity.
To arrive at this, we present two relativistic settings that, if embodied in special relativity, will never lead to any violation of the inequalities.
The first setting is general enough to cover all games presented here.
The second setting is simpler than the first, but inapplicable to the~$k$-cycle games for odd~$k$.
The reason for this is that the second setting requires the digraph to be two-colorable.

In a second part, we argue that if the settings are embodied in general relativity, then all presented inequalities may be violated.
Such a violation may be due to two characteristic features of general relativity:
(i) dynamic spacetime, i.e., the dependency of the spacetime structure from the distribution of matter, or
(2) the exotic and possibly fictitious feature of closed time-like curves (which are known to be compatible with general relativity~\cite{lanczos1924,godel1949,thorne1993,luminet2021}).\footnote{%
  In fact, classical-deterministic processes, such as those presented above (see, Definition~\ref{def:gwcm}) can in principle be implemented using ``tame'' closed time-like curves~\cite{baumeler2022a}.
}
Here, however, we disregard such causality violating descriptions~\cite{hawking1992}, and instead assume that the spacetime structure is free of such solutions.
Thus, a~violation of the presented inequalities would certify (1), the back-action of matter to the spacetime.

Crucial to the present endeavor is the notion of {\em event.}
It is central to understand that the Bell games presented provide limits on the correlations where the {\em events\/} form a partial order.
The quantum switch, mentioned in the introduction, has been demonstrated experimentally with quantum-optics tabletop experiments (see Ref.~\cite{goswami2020a} for a review).
Clearly, no significant gravitational effects entered these experiments.
An event, there, is understood as the reception and emission of a signal.
So, with that notion of event, our program fails.
Instead, we propose to use the commonly accepted notion of event in relativity {\em defined as the crossing of light beams, i.e., of null geodesics.}

\subsection{Special relativity}
In the following, we describe two settings where the relevant events in special relativity will {\em always\/} form a partial order.
The central ingredient is the same:
We define non-overlapping causal diamonds and convey the events to happen within these diamonds.
In special relativity---but in fact also in any curved spacetime, as long as it is fixed---the causal relations among the diamonds form a~partial order: the inequalities cannot be violated.
For our presentation, a single spatial dimension is sufficient (but not necessary).

\subsubsection{General setup}
Consider the three-agent setup schematically represented in Figure~\ref{fig:relativity}.
This setup is straight forwardly generalizable to any number of agents.
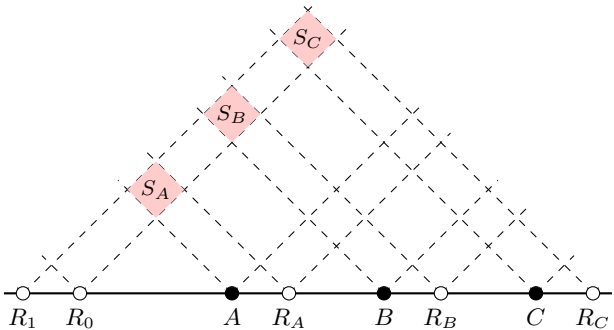
\begin{figure}
	\centering
	\begin{tikzpicture}
		\pgfmathsetmacro{\off}{0.75}
		\pgfmathsetmacro{\xrf}{0.25}
		\pgfmathsetmacro{\xr}{\xrf+\off}
		\pgfmathsetmacro{\xa}{3}
		\pgfmathsetmacro{\xb}{5}
		\pgfmathsetmacro{\xc}{7}
		\pgfmathsetmacro{\xra}{3+\off}
		\pgfmathsetmacro{\xrb}{5+\off}
		\pgfmathsetmacro{\xrc}{7+\off}
		\pgfmathsetmacro{\xw}{8}
		\draw[thick] (0,0) -- ++(\xw,0);
		\foreach \x in {\xrf,\xr,\xa,\xra,\xb,\xrb,\xc,\xrc} {
			\pgfmathsetmacro{\cleft}{\x/sqrt(2)}
			\pgfmathsetmacro{\cright}{(8-\x)/sqrt(2)}
			\draw[dashed] (\x,0) -- ++(45:\cright);
			\draw[dashed] (\x,0) -- ++(180-45:\cleft);
		}
		\node[ccc,label={below:$R_1$}] (R1) at (\xrf,0) {};
		\node[ccc,label={below:$R_0$}] (R0) at (\xr,0) {};
		\node[ccc,label={below:$R_A$}] (RA) at (\xra,0) {};
		\node[ccc,label={below:$R_B$}] (RB) at (\xrb,0) {};
		\node[ccc,label={below:$R_C$}] (RC) at (\xrc,0) {};
		\node[vertex,label={below:$A$}] (A) at (\xa,0) {};
		\node[vertex,label={below:$B$}] (B) at (\xb,0) {};
		\node[vertex,label={below:$C$}] (C) at (\xc,0) {};
		\pgfmathsetmacro{\h}{\off/(2*cos(45)}
		\pgfmathsetmacro{\xea}{\xr+(\xa-\xr)/2}
		\pgfmathsetmacro{\yea}{tan(45)*(\xea-\xr)}
		\path[fill=red!20!white] (\xea,\yea) -- ++(45:\h) -- ++(180-45:\h) -- ++(180+45:\h) -- cycle;
		\pgfmathsetmacro{\xeb}{\xr+(\xb-\xr)/2}
		\pgfmathsetmacro{\yeb}{tan(45)*(\xeb-\xr)}
		\path[fill=red!20!white] (\xeb,\yeb) -- ++(45:\h) -- ++(180-45:\h) -- ++(180+45:\h) -- cycle;
		\pgfmathsetmacro{\xec}{\xr+(\xc-\xr)/2}
		\pgfmathsetmacro{\yec}{tan(45)*(\xec-\xr)}
		\path[fill=red!20!white] (\xec,\yec) -- ++(45:\h) -- ++(180-45:\h) -- ++(180+45:\h) -- cycle;
		\pgfmathsetmacro{\yd}{sqrt(2*\h*\h)/2}
		\node at ($ (\xea,\yea) + (0,\yd) $) {\footnotesize$S_A$};
		\node at ($ (\xeb,\yeb) + (0,\yd) $) {\footnotesize$S_B$};
		\node at ($ (\xec,\yec) + (0,\yd) $) {\footnotesize$S_C$};
	\end{tikzpicture}
	\caption{The spacetime regions~$S_A,S_B,S_C$ form a partial order.
		Agent~$A$ may only obtain the input~$x$ and produce the output~$a$ in the region~$S_A$, and similarly for the agents~$B$ and~$C$.
		Therefore, the correlations~$p(a,b,c|x,y,z)$ decompose as a partial order, and no violation of the presented games can be observed.
	}
	\label{fig:relativity}
\end{figure}
Here, the agents~$A,B,C$ are initially spacelike separated, and so are the respective referees~$R_X$ for~$X\in\{A,B,C,0,1\}$.
In this setup, the referee~$R_0$ carries the inputs to the agents,
and the referees~$R_1,R_X$ for~$X\in\{A,B,C\}$, who travel at the speed of light, ensure that the outputs of the agents are timely produced.
For simplicity, we can imagine~$R_0$ to emit a light signal that carries the inputs.
The information encoded in this light signal must be encrypted in a way such that each agent cannot obtain the input to any other agent.
Also, we can imagine~$R_1$ and~$R_X$, for~$X\in\{A,B,C\}$, to emit light signals that switch off any detection device of agent~$X$.

By this setup, the input to each agent~$X\in\{A,B,C\}$ is only available in the intersection of the future lightcones of the initial location of~$X$ and~$R_0$.
If~$R_1$ and~$R_X$ meet without having received any output from agent~$X$, then the game is aborted.
This means that each agent~\mbox{$X\in\{A,B,C\}$} may only receive the input and produce the output in the spacetime region~$S_X$, {\it i.e.,} the event~$E_X$ is in
\begin{align}
	S_X:=J^+(R_0)\cap J^+(X) \setminus J^+(R_1)\cup J^+(R_X)
	\,,
\end{align}
where~$J^{+}(e)$ denotes the future lightcone of the event~$e$.
The spacetime regions~$\{S_X\}_{X\in\{A,B,C\}}$ are non-overlapping causal diamonds, and form a partial order.
Thus, the correlations~$p(\underline a|\underline x)$ attainable in this relativistic setting are partial-order correlations;
the presented Bell inequalities, by definition, cannot be violated.

\subsubsection{Two-colorable setup} 
Assume the digraph~$D=([2k],\mathcal A)$ that defines the game is two-colorable.
This is the case for the~$2k$-cycle digraph, and for the~$k$-fence and $k$-M\"obius digraphs.
Now, we bipartition the $2k$ vertices into the disjoint sets~$\mathcal V_L,\mathcal V_R$ according to the digraphs' two-coloring.
Any arc~$(s,r)\in\mathcal A$ announced by the referee will be such that~$s$ and~$r$ are of {\em different\/} color.
From this we learn that any communication among agents of the {\em same\/} color does not aid in winning the game.
This bipartition allows us to simplify the above setup.
In contrast to the first setup, where we define a causal diamond per event, we define a causal diamond per color:~$S_L$ for the vertices of color $L$, and~$S_R$ for the vertices of color~$R$.
In addition, we add a third causal diamond~$S_C$ in the causal past of~$S_L$ and~$S_R$.
This latter spacetime region will contain the event~$E_{n-1}$ of some extra agent~$n-1=2k$.
Without loss of generality, we may assume that~$S_L$ is spacelike separated form~$S_R$.
The causal arrangement of these diamonds is shown in Figure~\ref{fig:relativity2}.

\begin{figure}
	\centering
  \begin{tikzpicture}
    \tikzset{ccc/.style={vertex,fill=white}}
    \pgfmathsetmacro{\angle}{45}
    \pgfmathsetmacro{\xmax}{3.5}
    \draw[thick] (-\xmax,0) -- (\xmax,0);
    \pgfmathsetmacro{\off}{0.75}
    \pgfmathsetmacro{\xc}{0}
    \pgfmathsetmacro{\xcl}{-\off}
    \pgfmathsetmacro{\xcr}{+\off}
    \pgfmathsetmacro{\xlr}{-2}
    \pgfmathsetmacro{\xll}{\xlr-\off}
    \pgfmathsetmacro{\xrl}{+2}
    \pgfmathsetmacro{\xrr}{\xrl+\off}
    \pgfmathsetmacro{\cmax}{3}
    \foreach \x in {\xlr,\xc,\xrl} {
      \pgfmathsetmacro{\cleft}{min((\xmax+\x)/cos(\angle),\cmax)}
      \pgfmathsetmacro{\cright}{min((\xmax-\x)/cos(\angle),\cmax)}
      \draw[nulllike] (\x,0) -- ++(\angle:\cright);
      \draw[nulllike] (\x,0) -- ++(180-\angle:\cleft);
    }
    \foreach \x in {\xll,\xcl} {
      \pgfmathsetmacro{\cright}{min((\xmax-\x)/cos(\angle),\cmax)}
      \draw[->] (\x,0) -- ++(\angle:\cright);
    }
    \foreach \x in {\xrr,\xcr} {
      \pgfmathsetmacro{\cleft}{min((\xmax+\x)/cos(\angle),\cmax)}
      \draw[->] (\x,0) -- ++(180-\angle:\cleft);
    }
    \pgfmathsetmacro{\xright}{\xrl/2}
    \pgfmathsetmacro{\xleft}{-\xright}
    \pgfmathsetmacro{\xdelta}{\xright}
    \pgfmathsetmacro{\yup}{\xdelta*tan(\angle)}
    \pgfmathsetmacro{\h}{\off/(2*cos(\angle)}
    \pgfmathsetmacro{\yd}{sqrt(2*\h*\h)/2}
    \foreach \x/\y/\l in {\xleft/\yup/L,\xright/\yup/R,0/0/C} {
      \path[fill=red!20!white] (\x,\y) -- ++(\angle:\h) -- ++(180-\angle:\h) -- ++(180+\angle:\h) -- cycle;
      \node at ($ (\x,\y) + (0,\yd) $) {\footnotesize$S_\l$};
    }
    \foreach \x/\l in {\xc/C,\xlr/L,\xrl/R} {
      \node[ccc,label={below:$\An_{\l}$}] (An\l) at (\x,0) {};
    }
    \foreach \x/\l/\s in {\xll/L/\rightarrow,\xrr/R/\rightarrow,\xcl/R/\rightarrow,\xcr/L/\leftarrow} {
      \node[ccc,fill=black,label={below:$\Co^{\s}_{\l}$}] (R\l) at (\x,0) {};
    }
  \end{tikzpicture}
  \caption{%
    We define three non-overlapping causal diamonds.
    These diamonds are specified by the initial spatial configuration of the announcers ($\An$) and collectors ($\Co$).
  }
	\label{fig:relativity2}
\end{figure}

In order to ensure that the events happen within the respective spacetime regions, we propose the following setup.
Here, the referee is split into several announcers and collectors. 
The role of the announcers is to distribute the inputs to the agents.
The central announcer~$\An_C$ (see Figure~\ref{fig:relativity2}) broadcasts the arc~$(s,r)\in\mathcal A$, and a random bit~$x_\text{center}$.
The left announcer~$\An_L$ broadcasts a~random bit~$x_\text{left}$, the right announcer~$\An_R$ a random bit~$x_\text{right}$.
The semantics is the following.
Whenever~$s$ is of color $L$, then the random bit~$x$ to be guessed by~$r$ is~$x:=x_\text{center}\oplus x_\text{left}$, otherwise it is~$x:=x_\text{center}\oplus x_\text{right}$.
With this, the publicly announced arc~$(s,r)$ is accessible within~$J^+(\An_C)$.
The random bit~$x$, however, is only accessible within~$J^+(\An_C)\cap J^+(\An_\chi)$, where~$\chi\in\{L,R\}$ is the color of agent~$s$.
The role of the collectors, then again, is that the outputs of the agents are timely produced.
The collectors travel at speed of light in the direction indicated.
All agents of color~$\chi$ must produce their output the latest when~$\Co_\chi^\rightarrow$ and~$\Co_\chi^\leftarrow$ meet.
The extra agent~$2k$ must produce the output the latest when~$\Co_R^\rightarrow$ meets~$\Co_L^\leftarrow$.
If an agent did not produce an output by then, then the collectors abort the game.

These rules ensure that the events~$\{E_i\}_{i\in\mathcal V_\chi}$ are confined to the causal diamond
\begin{equation}
  S_\chi:=J^+(\An_C)\cap J^+(\An_\chi)\setminus J^+(\Co^\rightarrow_\chi)\cup J^+(\Co^\leftarrow_\chi)
  \,,
\end{equation}
and that~$E_{n-1}$ happens in the causal diamond
\begin{equation}
  S_C:=J^+(\An_C)\setminus J^+(\Co^\rightarrow_R)\cup J^+(\Co^\leftarrow_L)
  \,.
\end{equation}
The initial configuration, moreover, ensures that the causal diamonds are non-overlapping.
If these rules are enforced on top of special relativity, then the causal relations among the three causal diamonds form a partial order: The inequalities cannot be violated.

\subsection{General relativity}
As shown above (see Theorem~\ref{thm:winning}), the games presented are deterministically won with {\em causal\/} correlations.
This suggests that these games may also be won within general relativity (and likely with globally hyperbolic spacetime structures).
In general relativity, the distribution of matter determines the spacetime structure, and hence also the trajectories of particles and light.
In the suggested setups, the events are confined to causal diamonds.
The start and end points of these causal diamonds, then again, are defined via intersecting null geodesics.
Thus, since the distribution of matter influences the relative position of these intersections, it influences the causal relation between the causal diamonds.
With this, in particular, the causal relation between any relevant pair of events is determined by the distribution of matter within their past light cone.
Thus, one might imagine that it is in principle possible to arrange the distribution of matter in the events' common past to achieve either causal order.
This distribution, then again, could be arranged by an agent in the common past controlling the relevant pieces of matter.

For instance, consider the two-cycle game defined by the digraph~$D=([B,C],\{(B,C),(C,B)\})$, played with the three agents~$\{A,B,C\}$ arranged as in Figure~\ref{fig:relativity}.
In case the referee announces~$(B,C)$ and some~$x\in\{0,1\}$ to~$B$, then trivially,~$C$ may output~$x$.
Instead, if the referee announces~$(C,B)$ and some~$x\in\{0,1\}$ to~$C$, then in the special relativistic setting,~$B$ may at best guess~$x$ with half probability.
In a general-relativistic setting, however, agent~$A$ may alter the spacetime structure within its future lightcone:
Depending on the announced arc,~$(B,C)$ or~$(C,B)$, agent~$A$ may arrange the distribution of matter in such a way that the desired causal relation between the spacetime regions $S_B$ and $S_C$ is obtained.
In the second setup, the agent~$n-1=2k$ could in principle specify a distribution of matter such that the causal relation between events of different colors are determined by its will.
Here, we must also remark that general relativity is a deterministic theory.
The agent must thus be understood as deterministically and mechanistically carrying out a computation that displaces the matter distribution as desired.
Also, in order to arrive at the probabilities~$p(a|s,r,x)$, the experiment must be repeated over time with identical initial conditions, or repeated in parallel.

Note that whether such general-relativistic violations are feasible requires a formalization and a thorough analysis of the outlined strategies within the framework of general relativity. This remains a main challenge to our work and poses a central question among others.

\section{Conclusions and open questions}
We derived families of collaborative multiparty games that, if the winning chance exceeds the described limit, prove the parties' influence on the causal relations.
These games are formulated with directed graphs, and are rather simple:
A referee picks an arc~$(s,r)$, and asks party~$s$ to communicate a random bit to party~$r$.
The purpose of these games is to detect the dynamical spacetime structure present in general relativity and absent in special relativity.
This, however, can also be phrased purely within the theory of general relativity.
First, note that the spacetime structure in general relativity is said to be dynamical because matter exerts a back-action on it.
Thus, one can regard these Bell tests also as tests for the displacement, and therefore also for the presence, of matter.
Moreover, these games serve as device-independent test for {\em indefinite causal order,} as exhibited by the quantum switch.
It is known that the causal structure~$M$ of any unitary quantum process with indefinite causal order is cyclic~\cite{barrett2021nature}.
A candidate game to detect the dynamical component of the indefinite causal order is then simply obtained from the cyclic part of~$M$.

The game-winning strategies for the~$k$-cycle, the~$k$-fence, and the~$k$-M\"obius games require more parties than nodes in the graph.
It is possible, however, to design a game-winning causal model for the excluded case, as well as for the game in the ``all-to-one'' scenario (which is {\em nota bene\/} a causal game):
The Svetlichny-inspired causal model from Ref.~\cite{baumeler2020} has as causal structure the complete directed graph~$K^\text{di}_n$ among all parties.
By using that causal model, it is therefore trivial to communicate a bit from {\em any\/} party to {\em any other.}
The correlations that arise there, however, are incompatible with any global ordering of the parties, dynamical or not.

Our work raises a series of open questions, with the central one:
What is the precise general-relativistic description of the game-winning causal models?
The detection of dynamical causal structure can be understood as the detection of the curvature's change in the spacetime manifold.
Thus, an immediate follow-up question is whether and to what extent these games detect {\em gravitational waves.}
Speculatively, it might be possible to reinterpret the data collected at the LIGO experiments (see, {\it e.g.,} Ref.~\cite{ligo}) as a~violation of an inequality presented.

Understanding the informational content of general relativity might be beneficial for future research, especially when merging quantum theory with general relativity.
Towards that, not only answers to the above questions but also a mathematical formalization of the ``dynamics'' of causal structures might be helpful.
How do the experiments implemented in {\em local\/} regions alter the causal relation between events in their common future?

It is suspected that enumerating all facet-defining inequalities for the acyclic subdigraph polytope is infeasible:
The more nodes are involved, the more structure enters the polytope, and novel facets emerge~\cite{acyclic1985}.
In fact, deciding whether a vector is a member of the acyclic subdigraph polytope is NP-complete~\cite{garey1979}.
This raises two questions.
Firstly, inspired by the M\"obius inequality, one might wonder whether an orientation of the {\em Klein-bottle graph\/} would give rise to a facet-defining inequality or not.
After all, the Klein bottle is a generalization of the M\"obius strip.
Secondly, the computational difficulty mentioned above suggests that deciding whether some~$n$-party correlations are compatible with a partial order or not is NP-complete as well.
Note that this holds for local correlations in the context of quantum nonlocality~\cite{pitowsky1989}.

Finally, it is known~\cite{hardin2008} that the {\em axiom of choice\/} in set theory yields an advantage for the following game played among infinitely many players situated on a line:
Each player carries a hat with a random color, red or black, sees only the players in front, and must guess the own hat color.
This game, in fact, forces us to revise the concept of nonsignaling in physical theories~\cite{baumeler2022a}.
Now, the bounds on the games presented here are indifferent whether the players' actions form a partial or a {\em total order:}
We can equivalently assume the same configuration as for the ``hats'' game.
Does the axiom of choice also allow for an advantage for the ``infinite''-versions of the games presented?

\noindent
{\bf Acknowledgments.}
We thank Luca Apadula, Marios Christodoulou, Flavio Del Santo, Andrea Di Biagio, and Stefan Wolf for helpful discussions, and an anonymous referee for their valuable comments.
EET thanks Alexandra Elbakyan for providing access to the scientific literature.
EET is supported by the Austrian Science Fund (FWF) through ZK3 (Zukunftskolleg).
\"AB is supported by the Swiss National Science Foundation (SNF) through project~182452 and project 214808.

\bibliography{references.bib}
\end{document}